\documentclass[review]{elsarticle}

\usepackage{lineno,hyperref,color}
\modulolinenumbers[5]

\journal{}
\newtheorem{definition}{Definition}
\newtheorem{theorem}{Theorem}
\newtheorem{proof}{Proof}

\begin{document}

\begin{frontmatter}

\title{Application of the novel fractional grey model FAGMO(1,1,$k$) to predict China's
nuclear energy consumption \tnoteref{mytitlenote}}

\tnotetext[mytitlenote]{https://doi.org/10.1016/j.energy.2018.09.155}

 \author{Wenqing Wu \fnref{label1}}
 \address[label1]{School of Science, Southwest University of
                  Science and Technology, 621010, Mianyang, China}
 \author[label1]{Xin Ma \corref{cor1}\fnref{label2}}
 \cortext[cor1]{Corresponding author: Xin Ma}
 \ead{cauchy7203@gmail.com}
 \address[label2]{State Key Laboratory of Oil and Gas Reservoir Geology and Exploitation,
                  Southwest Petroleum University, 610500, Chengdu, China}
  \author{Bo Zeng \fnref{label3}}
  \address[label3]{College of Business Planning,
                   Chongqing Technology and Business
                   University, 400067, Chongqing, China}
                   \author{Yong Wang \fnref{label4,label2}}
  \address[label4]{School of Science, Southwest Petroleum
  University, 610500, Chengdu, China}
  \author{Wei Cai \fnref{label5}}
  \address[label5]{College of Engineering and Technology,
                   Southwest University, 400715, Chongqing, China}
\begin{abstract}
At present, the energy structure of China is shifting towards
cleaner and lower amounts of carbon fuel, driven by environmental needs and
technological advances. Nuclear energy, which is one of the major
low-carbon resources, plays a key role in China's clean energy
development. To formulate appropriate energy policies, it is necessary
to conduct reliable forecasts. This paper discusses the nuclear energy
consumption of China by means of a novel fractional grey model
FAGMO(1,1,$k$). The fractional accumulated generating matrix is
introduced to analyse the fractional grey model properties.
Thereafter, the modelling procedures of the FAGMO(1,1,$k$) are presented in
detail, along with the transforms of its optimal parameters. A
stochastic testing scheme is provided to validate the accuracy and
properties of the optimal parameters of the FAGMO(1,1,$k$). Finally, this model is used to forecast China's nuclear energy
consumption and the results demonstrate that the FAGMO(1,1,$k$) model
provides accurate prediction, outperforming other grey
models.
\end{abstract}
\end{frontmatter}

\section{Introduction}
\label{sec1:intro}

Energy is the most important strategic resource and provides a key
material basis for economic development and social progress. Energy
consumption prediction constitutes an important aspect of energy
policies for countries globally, particularly developing countries
such as China, where the energy consumption structure is changing at
a rapid speed. Numerous models have been introduced for forecasting
energy consumption, such as dynamic causality analysis
\cite{Mirza2017EnergyRSER}, nonlinear and asymmetric analysis
\cite{Shahbaz2017EnergyEE}, time-series analysis
\cite{Bekhet2017CO2RSEE, Amri2017TheRSEE}, machine learning models
\cite{fan2018comparison}, the coupling mathematical model
\cite{Wang2017TransientJPM,Wang2018FlowIJNSNS,Hu2018streamlineSPEPO},
autoregressive distributed lag model \cite{Brini2017RenewableRSEE},
hybrid forecasting system
\cite{Du2018MultistepRE,Cai2018developmentE}, machining system
\cite{Cai2018EnergyE}, fuzzy systems \cite{Wu2017TopologicalIJBC},
LEAP model \cite{schnaars1987houLRRP,Dong2017APS}, TIMES model
\cite{Shi2016modellingAE,Zhang2016timesAE}, NEMS model
\cite{Gabriel2001theOR, Soroush2017aIJSSD} and grey model
\cite{zeng2018forecasting,zeng2017self,Feng2012ForecastingESPBEPP,
Cui2013AAPM,duan2018forecasting,Zeng2018ForecastingE,Zeng2018improvedCIE,Wu2018usingE,Wang2018AnASC}.
Among these prevalent methods, simple linear regression,
multivariate linear regression, and time-series analysis are often
significant in accurately demonstrating the phenomena of long-term
trends. However, these exhibit the limitations of requiring a large
amount of observed data, at least 50 or more sets, to construct
models. The computational intelligence method requires a substantial
amount of training data to derive the optimised parameters. However,
in many practical situations, it is very difficult and sometimes
even impossible to obtain complete information. Therefore, it is
important to identify a favourable method for forecasting the trend
of an analysed system using scarce information with less errors.

The grey forecasting theory, proposed by Professor Deng
\cite{Deng1982ControlSCL}, offers a feasible and efficient method
for dealing with uncertain problems containing poor information. The
main advantage of this theory is that only four or more samples are
required to describe the behaviour and evolution of the analysed
system. In Deng's pioneering work, the first-order one variable grey
model GM(1,1) was discussed in detail. Over three decades of
development, the classical continuous GM(1,1) model has been studied
extensively; for example, by Xie et al.
\cite{Xie2009DiscreteAMM,Xie2013OnAPMAMM,Xie2017AJGS}, Wang et al.
\cite{Wang2017ForecastingJCP,Wang2017GreyAMM,Wang2017DecompositionEP},
Ma et al.
\cite{Ma2017ApplicationJCAAM,Ma2017ACNSNS,Ma2018ThekernelAMM,
Ma2018PredictingNCA} and others. However, we note that these
generalised grey models all include integer-order accumulation,
which results in less flexibility in time-series forecasting. Thus,
the fractional-order accumulation grey model is considered in this
paper.

By extending the integer accumulated generating operation into the
fractional accumulated generating operation, Wu et al.
\cite{Wu2013GreyCNSNS} first proposed the fractional accumulation
GM(1,1) model known as the FAGM(1,1) model. The computational
results demonstrated that the novel model outperformed the
conventional GM(1,1) model. Later, Wu and his peers successfully
applied fractional accumulation to the fuel production of China
\cite{Wu2014NonNCA}, tourism demand \cite{Wu2015UsingSC} and
electricity consumption \cite{Wu2015PropertiesAMAC}. Subsequently,
Xiao et al. \cite{Xiao2014TheAMM} studied the GM(1,1) model, in
which they regarded the fractional accumulated generator matrix as a
type of generalised accumulated generating operation. Gao et al.
\cite{Gao2015EstimationJGS} presented a new discrete fractional
accumulation GM(1,1) model known as FAGM(1,1,D) and applied it to
China's CO$_2$ emissions. Mao et al. \cite{Mao2016AAMM} investigated
a novel fractional grey model FGM($q$,1). Interested readers can
refer to \cite{Shen2014OptimizationJGS,
Yang2016ContinuousE,Yang2017ModifiedJAS,Liu2017NonJIM} for further
details on fractional accumulation grey models.

A further significant issue in grey system theory is that the
solution applied for prediction does not match the grey difference
equation. In 2009, Kong and Wei \cite{Kong2009OptimizationJGS}
proposed a parameter optimisation technique to study the DGM(2,1)
model. Later, Chen et al. used a similar technique to improve the
GM(1,1) \cite{Yu2017ApplyingJGS} and ONGM(1,1) models
\cite{Chen2014FoundationMPE}, in which the basic structure of the
original models remain in the optimised ones. Recently, Ma and Liu
\cite{Ma2016PredictingJCTN} studied the exact non-homogeneous grey
prediction model (ENGM) with an exact basic equation and background
value. Thereafter, Ma and Liu \cite{Ma2017TheJGS} considered the
GMC(1,$n$) model with optimised parameters and applied it to
forecasting the urban consumption per capita and industrial power
consumption of China. Following the concept of fractional
accumulation and the parameter optimisation method, we propose a
novel FAGMO(1,1,$k$) model.

In this paper, we study the nuclear energy consumption of China by
means of the FAGMO(1,1,$k$) model. The computational results
indicate that the proposed grey model outperforms the existing ENGM
model, optimised non-homogeneous grey model abbreviated as the
ONGM(1,1,$k$) model, FAGM(1,1) model and FAGM(1,1,$k$) model. The
main contributions of our paper are listed below. 1) A fractional
accumulation grey model with optimised parameters is developed. 2)
Detailed properties of optimised parameters are studied according to
two theorems. These indicate that the first parameter is the most
important factor affecting the accuracy of the FAGM(1,1,$k$) model.
3) Simulation results and two practical cases are considered to
assess the effectiveness of the FAGMO(1,1,$k$) model compared to
other models. 4) The FAGMO(1,1,$k$) grey forecasting model is
implemented to forecast the nuclear energy consumption of China. It
is demonstrated in the results that the newly proposed model offers
higher precision than other grey models.

The remainder of this paper is organised as follows. Section
\ref{sec:compendium} provides a compendium of China's energy
consumption. Section \ref{sec:Preliminaries} discusses several
preliminaries. A detailed discussion of the FAGM(1,1,$k$) model is
provided in section \ref{sec:FAGM}. Section \ref{sec:params} discusses
the optimised parameters. Modelling evaluation criteria and detailed
steps are provided in section \ref{sec:criteria}. Section
\ref{sec:validation} discusses the validation of the FAGMO(1,1,$k$)
model. Applications are explained in section \ref{sec:appli} and
conclusions are drawn in the final section.
 \section{Brief overview of China's energy
consumption} \label{sec:compendium}

This section presents a systematic and comprehensive investigation
of China's energy consumption using five fuels, namely coal, oil,
natural gas, nuclear energy and renewables. In China, renewables
include hydroelectricity, wind, solar, geothermal, biomass and
others. According to the statistical data of British Petroleum (BP)
{\it Statistical Review of World Energy 2018} (www.bp.com/
statisticalreview), the International Energy Agency (IEA) {\it World
Energy Outlook 2017} (www.iea.org/weo2017), Asia-Pacific Economic
Cooperation (APEC) {\it Energy Overview 2017}
(www.apec.org/Publications), and National Bureau of Statistics of
China (NBS) {\it China Statistical Yearbook 2017}
(www.stats.gov.cn/tjsj/ndsj), China's primary energy consumption
increased from 142.9 million tonnes oil equivalent (Mtoe) in the
first year of the third Five-Year Plan (1966 to 1970) to 3132.2 Mtoe
in the second year of the $13^{\rm th}$ Five-Year Plan (2016 to
2020), and increased dramatically since the turn of the millennium
owing to continuous economic growth. According to the statistical
data of BP, China's primary energy consumption from 1966 to 2017 is
plotted in Fig. \ref{fig:total-consumption}.
\begin{figure}[!htbp]
\centering
\includegraphics[height=7.5cm,width=11.5cm]{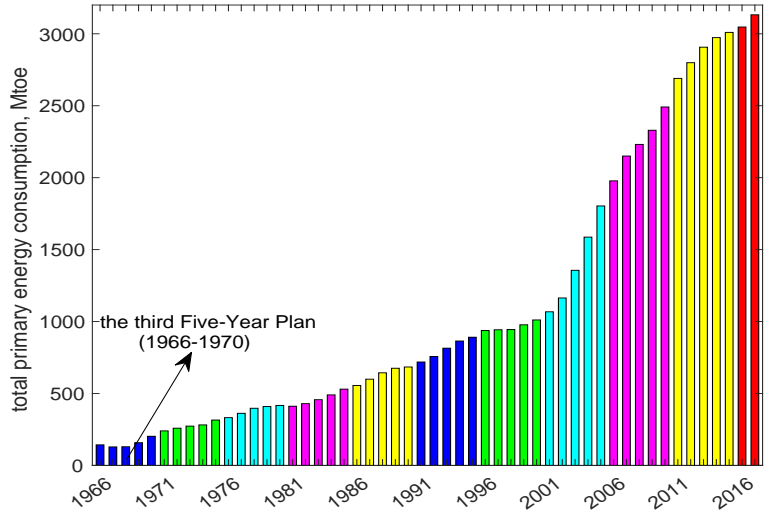}
\caption{Total primary energy consumption of China from 1966 to 2017}
 \label{fig:total-consumption}
\end{figure}

It is well known that China is the world's largest energy consumer,
accounting for 23\% and 23.2\% of the global energy consumption in
2016 and 2017, respectively. While coal remains the dominant fuel,
its share of total energy consumption was 62\% in 2016 and 60.4\% in
2017. China's $13^{\rm th}$ Five-Year Plan set an ambitious target
for adjusting the primary energy consumption structure. The energy
plan set by China for the $13^{\rm th}$ Five-Year Plan can met the
adjustment target of the primary energy consumption structure. A
brief overview of China's primary energy consumption from the
perspective of five fuel types is provided below.
\subsection{Coal}
\label{subsec:coal}

Since the foundation of the People's Republic of China, coal has
always been the primary energy fuel, owing to abundant domestic
reserves and its low cost \cite{Dong2017APS}. From Figs.
\ref{fig:value-consumption} and \ref{fig:percentage-consumption}, it
can be observed that coal soared from 122.4 Mtoe in 1966 to 1892.6
Mtoe in 2017, although the percentage of coal in the total primary
energy consumption decreased from 85.7\% in 1966 to 60.4\% in 2017.
Specifically, despite a continuous increase in coal consumption
during the third and fourth Five-Year Plan periods, the proportion
of coal in the total primary energy consumption has decreased from
85.7\% to 72.5\%. Following China's Reform and Opening-Up Policy in
1978, the coal consumption has expanded rapidly from 282.8 Mtoe in
1978 to 1685.8 Mtoe in 2010, while the share of coal in the total
primary energy consumption is around at 73.8\%. At the beginning of
the $12^{\rm th}$ Five-Year Plan (2011 to 2015), the Chinese
government has been stepping up its efforts to reduce coal
consumption to deal with air pollution and climate change. The
``supply-side reform" removes unnecessary and out-dated production
capacity to avoid supply overcapacity in the coal mining industry.
During this period, the consumption of coal decreased from 1903.9
Mtoe in 2011 to 1892.6 Mtoe in 2017, while the share of coal
decreased from 70.8\% to 60.4\%.

\subsection{Oil}
\label{subsec:oil}

Oil is a major component of primary energy resources globally and
plays a strategic role in economic growth. It can be observed in
Figs. \ref{fig:value-consumption} and
\ref{fig:percentage-consumption} that the oil consumption in China
increased from 14.3 Mtoe in 1966 to 608.4 Mtoe in 2017, with an
average annual growth rate of 7.5\%. Owing to China's oil reserves
accounting for only 2\% of the global amount, China is highly
dependent on overseas oil imports of more than 60\%. China became a
net importer of crude oil in 1993 and the world's second largest oil
consumer in 2002 \cite{Dong2017APS}. In April 2015, China surpassed
the US as the world's largest oil importer, with imports of 7.4
million barrels per day (Mbbl/D), thereby exceeding the US imports
of 7.2 Mbbl/D. Following the $11^{\rm th}$ Five-Year Plan, the oil
consumption increased rapidly owing to economic growth and the
improved quality of life.
\begin{figure}[!htbp]
\centering
\includegraphics[height=7.5cm,width=11cm]{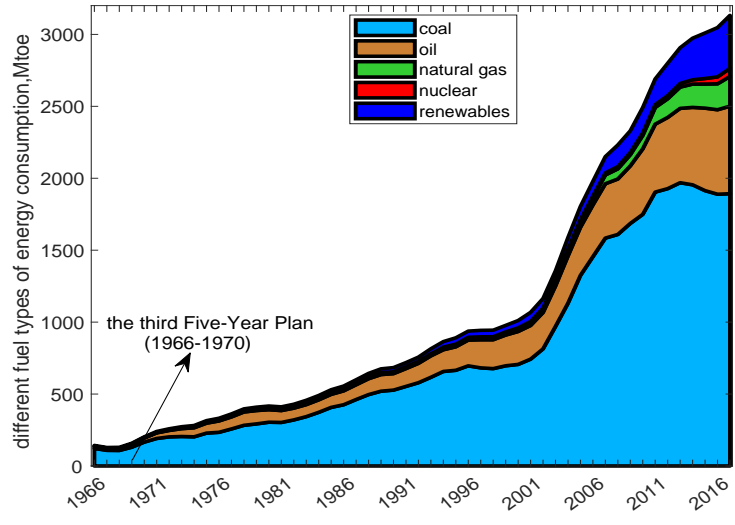}
\caption{China's primary energy consumption under fuel types}
 \label{fig:value-consumption}
\end{figure}
\begin{figure}[!htbp]
\centering
\includegraphics[height=7.5cm,width=11cm]{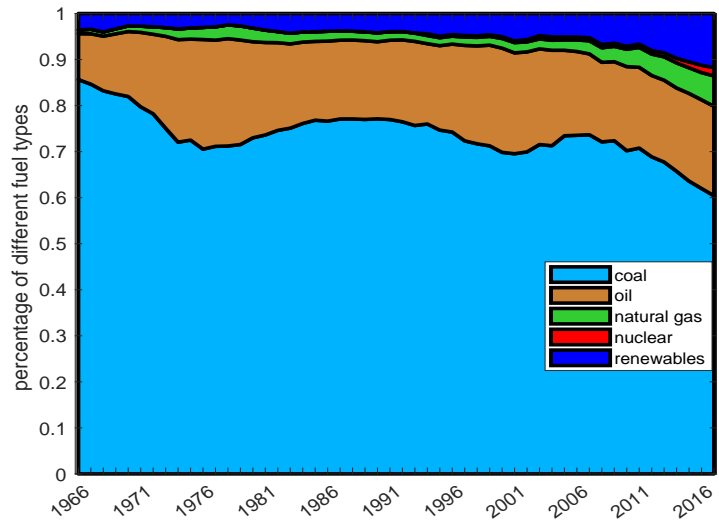}
\caption{Percentage of China's primary energy consumption under fuel
types}
 \label{fig:percentage-consumption}
\end{figure}

\subsection{Natural gas}
\label{subsec:gas}

Natural gas is a fossil fuel for electricity generation, chemical
feedstock, heating and cooking, among others. Chinese organisations
have estimated that the technically and ultimately recoverable
resources of natural gas are 6.1 trillion cubic meters (tcm) and 37
tcm \cite{Hou2015unconventionalEES}, respectively. However, natural
gas has not become a major energy resource in China because the
domestic natural gas industry has developed slowly. In recent years,
the Chinese government has set the stable natural gas supply as one
of the country's energy strategies and encourages gas transportation
from areas with significant resources to East China. The National
Development and Reform Commission constructed three west-east gas
pipelines in 2004, 2007 and 2015, respectively. Furthermore, the
``shifting from coal to gas" policy has a significant impact on the
natural gas market. The natural gas consumption has increased from
116.2 Mtoe in 2011 to 206.7 Mtoe in 2017, with an average annual
growth rate of 8.6\%.

\subsection{Nuclear energy}
\label{subsec:nuclear}

Nuclear energy is almost always used to generate electricity. To
reduce the air pollution from coal-fired power plants, nuclear
energy is an inevitable strategic option for China. In fact, China
began to develop nuclear energy in the 1980s and the Qinshan Nuclear
Power Plant began operating in 1991. In 2012, the State Council set
a goal of 58 GW nuclear capacity by 2020. At the beginning of the
$13^{\rm th}$ Five-Year Plan, 38 nuclear power reactors were in
operation with a production of 213.3 TWh, while 19 nuclear power
reactors were under construction. At present, the Chinese government
focuses on fourth-generation reactors with increased safety. From
the $11^{\rm th}$ Five-Year Plan (2006 to2010), nuclear energy
consumption has soared rapidly from 12.4 Mtoe in 2006 to 56.2 Mtoe
in 2017, with an average annual growth rate of 13.4\%.

\subsection{Renewables}
\label{subsec:renawables}

China's renewable energy has been expanding rapidly in recent
decades, owing to the development of the modern renewable energy
industry. In 2017, China's renewables consumption accounted for
21.9\% of the total global amount, increasing by 31\% and accounting
for 36\% of the global renewables consumption growth. Meanwhile, the
renewables consumption increased from 101.1 Mtoe in 2006 to 368.3
Mtoe in 2017; the share has increased from 5.1\% to 11.8\% with an
average annual growth rate of 11.4\%. The $13^{\rm th}$ Five-Year
Plan set targets for an installed wind power generation capacity of
250 GW, solar power generation capacity of 110 GW, and hydropower
generation capacity of 350 GW by 2020.

In summary, China's primary energy consumption using five fuels for
the period of 1966 to 2017 can be provided below. The coal consumption
has gradually declined, the oil consumption has gradually increased,
and the natural gas, nuclear energy and renewables have rapidly
increased. China's primary energy consumption structure exhibits a
diversified trend, and the clean energy has increased yearly.


\section{Definitions and properties of fractional accumulation}
\label{sec:Preliminaries}

This section provides the fractional accumulated generating
operation (AGO), which can reduce the randomness of raw data in grey
theory. Correspondingly, the inverse operation of accumulated
generation is known as the inverse accumulated generating operation
(IAGO). The $r^{\rm th}$ AGO and $r^{\rm th}$ IAGO are provided
below, which can be found in paper
\cite{Wu2013GreyCNSNS,Mao2016AAMM}.
\begin{definition}\label{def:ago}
Let $X^{\left( 0 \right)}  = \left\{ {x^{\left( 0 \right)}
(1),x^{\left( 0 \right)} (2), \ldots,x^{\left( 0 \right)} (n)}
\right\}$ be an original sequence and $X^{\left( r \right)} \left(
{r > 0} \right)$ be the $r^{\rm th}$ accumulated generating
operation ($r$-AGO) sequence of $X^{\left( 0 \right)}$, where
$x^{\left( r \right)} (k) = \sum\limits_{i = 1}^k {x^{\left( {r - 1}
\right)} (i)},k = 1,2, \ldots, n$. Denote by $A^r$ the $r$-AGO
matrix that satisfies $X^{\left( r \right)} = X^{\left( 0 \right)}
A^r$, and
 \begin{eqnarray}
A^r = \left( {\begin{array}{*{20}c}
          \left[ {_{0}^{r} } \right]
      &   \left[ {_{1}^{r} } \right]
      &   \left[ {_{2}^{r} } \right]
     &  \cdots
     & \left[ {_{n - 1}^{\left.{\rule{4pt}{0mm}}\right. r} } \right] \\   
               0
     &  \left[ {_{0}^{r} } \right]
     &  \left[ {_{1}^{r} } \right]
     &      \cdots
     & \left[ {_{n - 2}^{\left.{\rule{4pt}{0mm}}\right. r} } \right] \\ 
             0
     &       0
     & \left[ {_{0}^{r} } \right]
     &  \cdots
     & \left[ {_{n - 3}^{\left.{\rule{4pt}{0mm}}\right. r} } \right] \\ 
    \vdots  &  \vdots  &  \vdots  &  \ddots  &  \vdots   \\ 
      0 &    0 &    0    &  \cdots
      & \left[ {_{0}^{r} } \right] \\ 
     \end{array}} \right)_{n \times n}, \nonumber
\end{eqnarray}
with $ \left[ {_{i}^{r} } \right]
  = \frac{{r \left( {r + 1} \right) \cdots
        \left( {r + i - 1} \right)}}{{i!}}
        ={r+i-1 \choose i}=\frac{(r+i-1)!}{i!(r-1)!},\
 \left[ {_{i}^{0} } \right]= 0,\ \left[ {_{0}^{0} } \right]
        = {0 \choose 0}= 1. $
\end{definition}

Obviously, the 1-AGO sequence $x^{( 1)} (k) = \sum\limits_{i = 1}^k
{x^{( 0 )} (i)}, k = 1, 2, \ldots, n$, namely $X^{\left( 1 \right)}
= X^{\left( 0 \right)} A$ with $A = \left( {\begin{array}{*{20}c}
   1 & 1 & 1 &  \cdots  & 1  \\
   0 & 1 & 1 &  \cdots  & 1  \\
   0 & 0 & 1 &  \cdots  & 1  \\
    \vdots  &  \vdots  &  \vdots  &  \ddots  &  \vdots   \\
   0 & 0 & 0 &  \cdots  & 1  \\
\end{array}} \right)_{n \times n}$.
\begin{definition}\label{def:iago}
The inverse accumulated generation is defined as $x^{\left( {r - 1}
\right)} (k) = x^{\left( r \right)} (k) - x^{\left( r \right)} (k -
1),k = 1,2, \ldots,n$. Denote by $D^r$ the $r^{\rm th}$ inverse
accumulated generating operation ($r$-IAGO) matrix, which satisfies
$X^{\left( 0 \right)} = X^{\left( r \right)} D^r$, and
\begin{eqnarray}
D^r  = \left( {\begin{array}{*{20}c}
            \left[ {^{-r}_{\left.{\rule{2pt}{0mm}}\right. 0} } \right]
      &     \left[ {^{-r}_{\left.{\rule{2pt}{0mm}}\right. 1} } \right]
      &     \left[ {^{-r}_{\left.{\rule{2pt}{0mm}}\right. 2} } \right] &  \cdots
     &  \left[ {_{n - 1}^{\left.{\rule{1.5pt}{0mm}}\right. -r} } \right] \\ 
     0  &  \left[ {^{-r}_{\left.{\rule{2pt}{0mm}}\right. 0} } \right]
        &  \left[ {^{-r}_{\left.{\rule{2pt}{0mm}}\right. 1} } \right]
        &  \cdots  & \left[ {_{n - 2}^{\left.{\rule{1.5pt}{0mm}}\right. -r} } \right] \\ 
    0 & 0 & \left[ {^{-r}_{\left.{\rule{2pt}{0mm}}\right. 0} } \right]   &  \cdots
    &  \left[ {_{n - 3}^{\left.{\rule{1.5pt}{0mm}}\right. -r} } \right] \\  
    \vdots  &  \vdots  &  \vdots  &  \ddots  &  \vdots   \\  
    0 & 0 & 0 &  \cdots  & \left[ {^{-r}_{\left.{\rule{2pt}{0mm}}\right. 0} } \right] \\  
      \end{array}} \right)_{n \times n}, \nonumber
\end{eqnarray}
with $\left[ {^{-r}_{\left.{\rule{2pt}{0mm}}\right. i} } \right]
     \!=\! \frac{{ - r \left( { - r + 1} \right) \cdots
           \left( { - r + i - 1} \right)}}{{i!}}
  \!=\! \left( { - 1} \right)^i
     \frac{{r \left( {r - 1} \right) \cdots
      \left( {r - i + 1} \right)}}{{i!}}
  \!\!=\!\! \left( { - 1} \right)^i {r\choose i},
 \left[ {^{-r}_{\left.{\rule{2pt}{0mm}}\right. i} } \right]\!= \!0, i\! > r. $

\end{definition}

Similarly, the 1-IAGO sequence $x^{\left( 0 \right)} (k) = x^{\left(
1 \right)} (k) - x^{\left( 1 \right)} (k - 1),k = 1,2, \ldots,n$;
that is, $X^{\left( 0 \right)} = X^{\left( 1 \right)} D$, with $D =
\left( {\begin{array}{*{20}c}
   1 & { - 1} & 0 &  \cdots  & 0  \\
   0 & 1 & { - 1} &  \cdots  & 0  \\
   0 & 0 & 1 &  \cdots  & 0  \\
    \vdots  &  \vdots  &  \vdots  &  \ddots  &  \vdots   \\
   0 & 0 & 0 &  \cdots  & 1  \\
\end{array}} \right)_{n \times n}$.

 \begin{theorem}
 \label{Th:r-qu-i}
The expression $\left[ {_i^r } \right],r \in R^ +,i \in N^ +$ is a
function of $r$ and $i$; for any value $i$,

 $r \in \left( {0,1} \right),\ \left[ {_i^r }
\right]$ is a monotonically decreasing function of $i$;

 $r = 1,\ \left[ {_i^r } \right] \equiv 1$; and

 $r \in \left( {1,+\infty} \right),\ \left[ {_i^r } \right]$ is a
monotonically increasing function of $i$.
\end{theorem}

\begin{proof} We consider the difference
\begin{eqnarray}
&&\left[ {_i^r } \right]
 - \left[ {_{i -1}^{\left.{\rule{2pt}{0mm}}\right. r} } \right]
  = \left( {^{r+i-1}_{\left.{\rule{7pt}{0mm}}\right. i} } \right)
       - \left( {^{r+i-2}_{\left.{\rule{2pt}{0mm}}\right. i-1} } \right)
 = \frac{{\left( {r + i - 1} \right)!}}{{i!\left( {r - 1} \right)!}}
 - \frac{{\left( {r + i - 2} \right)!}}{{\left( {i - 1} \right)!\left( {r - 1} \right)!}} \nonumber\\
&& \hskip1.7cm= \frac{{\left( {r + i - 2} \right)!}}{{\left( {i - 1}
\right)!\left( {r - 1} \right)!}}
    \left[ {\frac{{r + i - 1}}{i} - 1} \right] \nonumber\\
&& \hskip1.7cm = \frac{{\left( {r + i - 2} \right)!}}{{\left( {i -
1} \right)!
   \left( {r - 1} \right)!}}\frac{{r - 1}}{i}\nonumber\\
&&\hskip1.7cm = \frac{{\left( {r + i - 2} \right)!}}{{ i!
   \left( {r - 1} \right)!}} (r - 1).\nonumber
\end{eqnarray}

From the difference results, we complete the proof.
\end{proof}

To gain an improved understanding of Theorem \ref{Th:r-qu-i}, two
figures are displayed in the following Fig. \ref{fig:r-three}

\begin{figure}[!htbp]
\centering
\begin{minipage}[c]{0.5\textwidth}
\centering
\includegraphics[height=5.5cm,width=6cm]{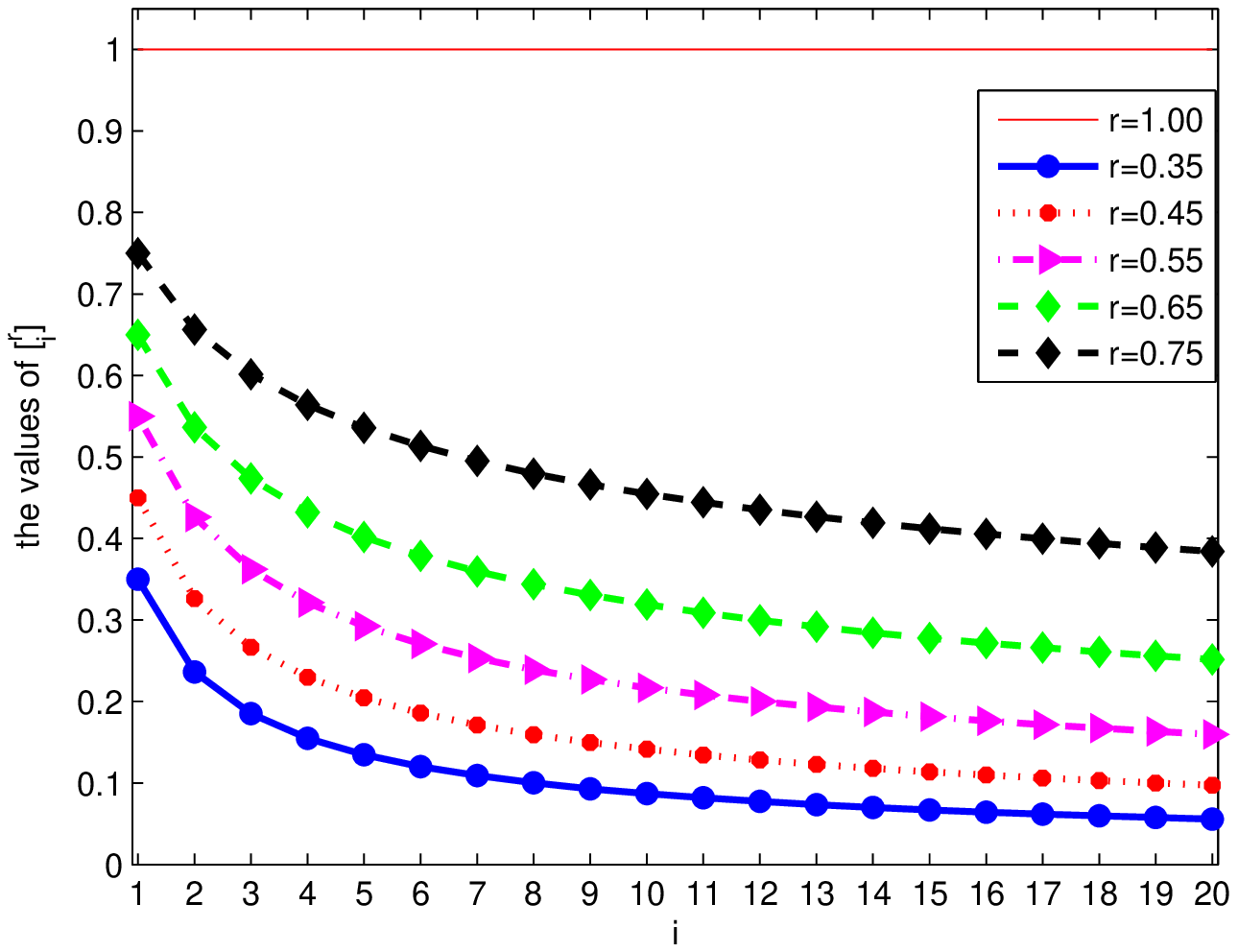}
\end{minipage}%
\begin{minipage}[c]{0.5\textwidth}
\centering
\includegraphics[height=5.5cm,width=6cm]{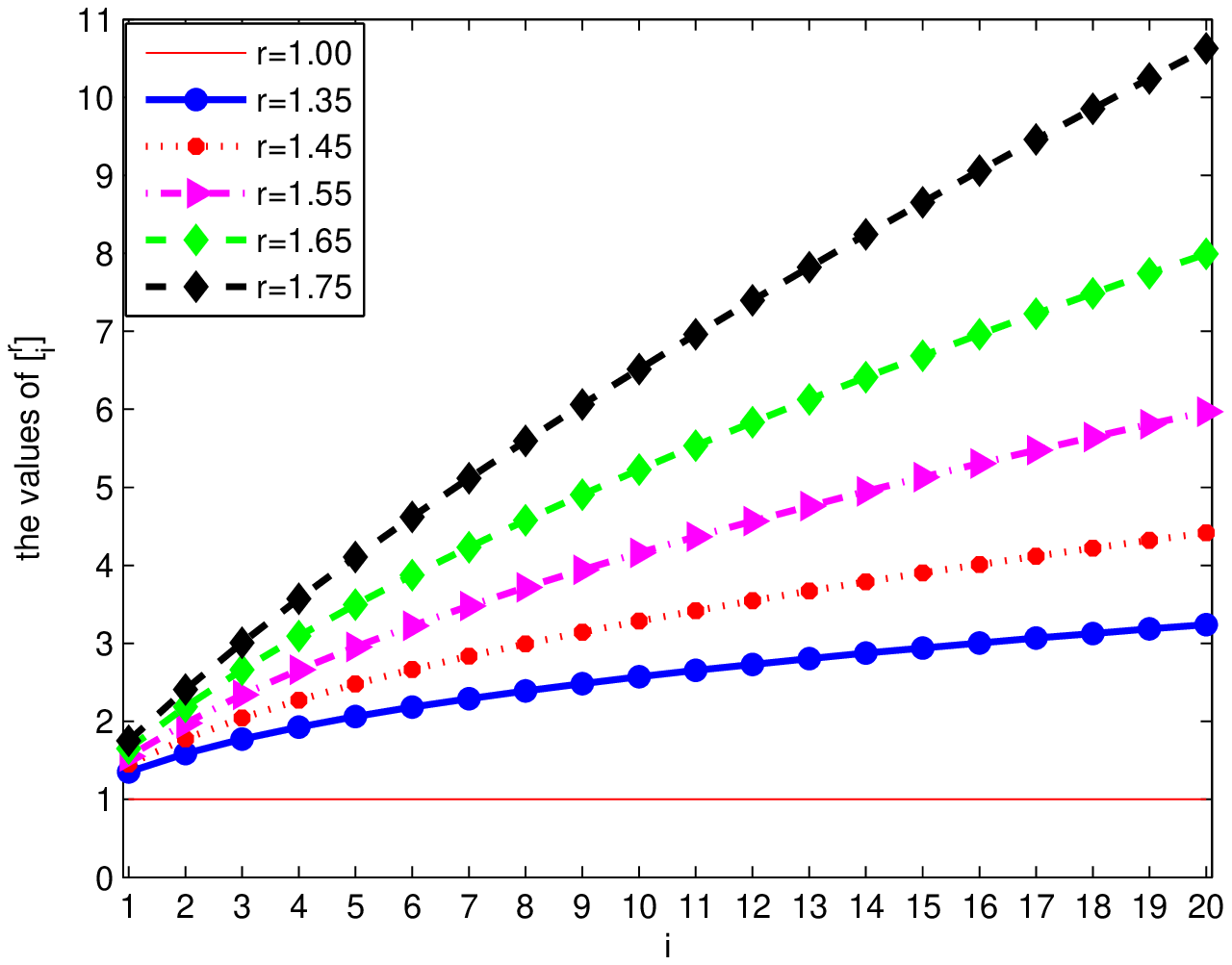}
\end{minipage}
\caption{Function $\left[ {_i^r } \right]$  versus values
$r$ and $i$: left $r \in \left( {0,1} \right)$, right $r \in \left(
{1, +\infty} \right)$}
  \label{fig:r-three}
\end{figure}

It follows from $X^{\left( r \right)}  = X^{\left( 0 \right)} A^r$
that
 \begin{eqnarray}
  x^{\left( r \right)} \left( k \right)
   = \sum\limits_{i =1}^k
   \left[ {_{k - i}^{\left.{\rule{2pt}{0mm}}\right. r} } \right]
    x^{\left( 0 \right)} (i)
   = \sum\limits_{i = 0}^{k-1} {\left[ {_i^r } \right]}
    x^{\left( 0 \right)} (k - i),
 \end{eqnarray}
which means that $x^{\left( r \right)} \left( k \right)$ is the weight of
$x^{\left( 0 \right)} (i),i = 1,2, \ldots,k$.

From Theorem \ref{Th:r-qu-i}, when $r \in \left( {0,1} \right)$, the
weight of the old data is smaller than that of the new data. When
$r=1$, the weights of the old and new data are all 1. When
$r \in \left( {1, +\infty} \right)$, the weight of the old data is
larger than that of the new data.

\begin{theorem}\label{th:ago-and-iago}
 The values of $r$-AGO $A^r$ and $r$-IAGO $D^r$ satisfy $\left( {A^r } \right)^{ - 1}  = D^r$.
\end{theorem}

\begin{proof}
From the definition of $A^r$, it is easy to calculate the determinant
$det\left( {A^r } \right) = 1$, which means that $A^r$ is reversible.

Employing mathematical induction, when $r=1$, we obtain
\begin{eqnarray}
AD = \left( {\begin{array}{*{20}c}
   1 & 1 & 1 &  \cdots  & 1  \\
   0 & 1 & 1 &  \cdots  & 1  \\
   0 & 0 & 1 &  \cdots  & 1  \\
    \vdots  &  \vdots  &  \vdots  &  \ddots  &  \vdots   \\
   0 & 0 & 0 &  \cdots  & 1  \\
\end{array}} \right)\left( {\begin{array}{*{20}c}
   1 & { - 1} & 0 &  \cdots  & 0  \\
   0 & 1 & { - 1} &  \cdots  & 0  \\
   0 & 0 & 1 &  \cdots  & 0  \\
    \vdots  &  \vdots  &  \vdots  &  \ddots  &  \vdots   \\
   0 & 0 & 0 &  \cdots  & 1  \\
\end{array}} \right) =  I. \nonumber
\end{eqnarray}

Assuming that the properties hold true when $r=m$, this means that
\begin{eqnarray}
 \left( {\begin{array}{*{20}c}
   \left[ {^{m}_{\left.{\rule{0.1pt}{0mm}}\right. 0} } \right]
    & \left[ {^{m}_{\left.{\rule{0.1pt}{0mm}}\right. 1} } \right]
    &\left[ {^{m}_{\left.{\rule{0.08pt}{0mm}}\right. 2} } \right]&  \cdots
    & \left[ {_{n -1}^{\left.{\rule{2pt}{0mm}}\right. m} } \right]  \\
   0 & \left[ {^{m}_{\left.{\rule{0.1pt}{0mm}}\right. 0} } \right]
   & \left[ {^{m}_{\left.{\rule{0.1pt}{0mm}}\right. 1} } \right] &  \cdots
   & \left[ {_{n -2}^{\left.{\rule{2pt}{0mm}}\right. m} } \right]  \\
   0 & 0 & \left[ {^{m}_{\left.{\rule{0.1pt}{0mm}}\right. 0} } \right] &  \cdots
   & \left[ {_{n -3}^{\left.{\rule{2pt}{0mm}}\right. m} } \right]  \\
    \vdots  &  \vdots  &  \vdots  &  \ddots  &  \vdots   \\
   0 & 0 & 0 &  \cdots  & \left[ {^{m}_{\left.{\rule{0.1pt}{0mm}}\right. 0} } \right]  \\
\end{array}} \right)\left(\! {\begin{array}{*{20}c}
   \left[ {^{-m}_{\left.{\rule{4pt}{0mm}}\right. 0} } \right]
   & \left[ {^{-m}_{\left.{\rule{4pt}{0mm}}\right. 1} } \right]
   & \left[ {^{-m}_{\left.{\rule{4pt}{0mm}}\right. 2} } \right] &  \cdots  & {\left[ {_{n - 1}^{ - m} } \right]}  \\
   0 & \left[ {^{-m}_{\left.{\rule{4pt}{0mm}}\right. 0} } \right]
   & \left[ {^{-m}_{\left.{\rule{4pt}{0mm}}\right. 1} } \right] &  \cdots  & {\left[ {_{n - 2}^{ - m} } \right]}  \\
   0 & 0 & \left[ {^{-m}_{\left.{\rule{4pt}{0mm}}\right. 0} } \right] &  \cdots  & {\left[ {_{n - 3}^{ - m} } \right]}  \\
    \vdots  &  \vdots  &  \vdots  &  \ddots  &  \vdots   \\
   0 & 0 & 0 &  \cdots  & \left[ {^{-m}_{\left.{\rule{4pt}{0mm}}\right. 0} } \right]  \\
\end{array}} \! \right)\! =\! I. \nonumber
\end{eqnarray}

Then, when $r=m+1$, we obtain
\begin{eqnarray}
A^{m + 1} D^{m + 1}  = A^m (AD)D^m  = A^m ID^m  = A^m D^m  = I,
\nonumber
\end{eqnarray}
so the result $\left( {A^r } \right)^{ - 1}  = D^r$ is proven.
\end{proof}
%
\section{Fractional grey FAGM(1,1,$k$) model}
\label{sec:FAGM}

\begin{definition}
 \label{def:FAGM11k}
The first-order differential equation
\begin{eqnarray}
 \label{Eq:fAGM11k}
\frac{{dx^{\left( r \right)} (t)}}{{dt}} + ax^{\left( r \right)} (t)
= bt + c,\ r > 0
\end{eqnarray}
is known as the whitening differential equation of the FAGM(1,1,$k$)
model. The parameter $a$ is a development coefficient, while $b t+c$ is
the grey action quantity.
\end{definition}

The discrete differential equation
\begin{eqnarray}
 \label{Eq:differential}
x^{\left( {r - 1} \right)} (k) + az^{\left( r \right)} (k) =
b\frac{{2k - 1}}{2} + c
\end{eqnarray}
is referred to as the basic equation of the FAGM(1,1,$k$). $x^{\left( {r -
1} \right)} (k) = x^{\left( r \right)} (k) - x^{\left( r \right)} (k
- 1)$, $z^{\left( r \right)} (k) = 0.5\left( {x^{\left( r \right)}
(k - 1) + x^{\left( r \right)} (k)} \right)$.

The least-squares estimation for $\phi  = \left( {a,b,c} \right)$ of
the FAGM(1,1,$k$) model satisfies
\begin{eqnarray}
 \label{Eq:esparams}
\phi  = \left( {B^\mathcal{T} B} \right)^{ - 1} B^\mathcal{T} Y,
\end{eqnarray}
where
\begin{eqnarray}
B = \left( {\begin{array}{*{20}c}
   { - z^{\left( r \right)} (2)} & {\frac{3}{2}} & 1  \\
   { - z^{\left( r \right)} (3)} & {\frac{5}{2}} & 1  \\
    \vdots  &  \vdots  &  \vdots   \\
   { - z^{\left( r \right)} (\nu)} & {\frac{{2\nu - 1}}{2}} & 1  \\
\end{array}} \right), \quad Y = \left( {\begin{array}{*{20}c}
   {x^{\left( {r - 1} \right)} (2)}  \\
   {x^{\left( {r - 1} \right)} (3)}  \\
    \vdots   \\
   {x^{\left( {r - 1} \right)} (\nu)}  \\
\end{array}} \right), \nonumber
\end{eqnarray}
in which $\nu$ is the number of samples used to construct the model.

\begin{theorem}
 \label{Th:responsefunction}
The time response function of the FAGM(1,1,$k$) model is
\begin{eqnarray}
 \label{Eq:responsefunction}
\hat x^{\left( r \right)} (k) = \left[ {x^{\left( 0 \right)} (1) -
\frac{b}{a} + \frac{b}{{a^2 }} - \frac{c}{a}} \right]e^{ - a(k - 1)}
+ \frac{b}{a}k - \frac{b}{{a^2 }} + \frac{c}{a},\ k = 2,3, \ldots
,n,
\end{eqnarray}
and the restored value of $\hat x^{\left( 0 \right)} (k)\ k = 2,3,
\ldots,n$ can be expressed by
\begin{eqnarray}
 \label{Eq:restoredvalues}
\hat X^{\left( 0 \right)}  = \hat X^{\left( r \right)} D^r.
\end{eqnarray}
\end{theorem}

\begin{proof}
From Eq. (\ref{Eq:fAGM11k}), we have
\begin{eqnarray}
 \label{Eq:fromEq1}
\frac{{dx^{\left( r \right)} (t)}}{{dt}} =  - ax^{\left( r \right)}
(t) + b t + c.
\end{eqnarray}

Let $u(t) =  - ax^{\left( r \right)} (t) + bt + c$; then, Eq.
(\ref{Eq:fromEq1}) is transformed into
\begin{eqnarray}\label{Eq7:fromEq6}
 \frac{{du(t)}}{{dt}} =- a\frac{{dx^{\left( r \right)}
(t)}}{{dt}} +b= - au(t) + b.
\end{eqnarray}

To perform the indefinite integral on Eq. (\ref{Eq7:fromEq6}) and
reduce it, we obtain
\begin{eqnarray}\label{Eq8:fromEq7}
 - a\left( { - ax^{\left( r \right)} (t) + bt + c} \right) + b = e^{ - at}
 e^\kappa,
\end{eqnarray}
where $e^\kappa$ is a constant to be determined.

Substituting $t=1$ and $\left. {x^{\left( r \right)} (t)} \right|_{t
= 1}  = x^{\left( 0 \right)} (1)$ into Eq. (\ref{Eq8:fromEq7}), we
obtain
\begin{eqnarray}\label{Eq9:ek}
e^\kappa   = e^a \left( {a^2 x^{\left( 0 \right)} (1) - ab - ac + b}
\right).
\end{eqnarray}

It follows from Eqs. (\ref{Eq8:fromEq7}) and (\ref{Eq9:ek}) that
\begin{eqnarray}\label{Eq10: xtformula}
x^{\left( r \right)} (t) = \left[ {x^{\left( 0 \right)} (1) -
\frac{b}{a} + \frac{b}{{a^2 }} - \frac{c}{a}} \right]e^{ - a(t - 1)}
+ \frac{b}{a}t - \frac{b}{{a^2 }} + \frac{c}{a}.
\end{eqnarray}

Thus, the time response function of the FAGM(1,1,$k$) model is
\begin{eqnarray}
\hat x^{\left( r \right)} (k) = \left[ {x^{\left( 0 \right)} (1) -
\frac{b}{a} + \frac{b}{{a^2 }} - \frac{c}{a}} \right]e^{ - a(k - 1)}
+ \frac{b}{a}k - \frac{b}{{a^2 }} + \frac{c}{a},k = 2,3, \ldots,n,
\nonumber
\end{eqnarray}
and the restored value of $\hat x^{\left( 0 \right)} (k)\ k = 2,3,
\ldots,n$ can be expressed by
\begin{eqnarray}
\left( {\hat x^{\left( 0 \right)} (1),\hat x^{\left( 0 \right)} (2),
\ldots,\hat x^{\left( 0 \right)} (n)} \right) = \left( {\hat
x^{\left( r \right)} (1),\hat x^{\left( r \right)} (2), \ldots,\hat
x^{\left( r \right)} (n)} \right)D^r.\nonumber
\end{eqnarray}
\end{proof}

Setting $b=0$ in Eq. (\ref{Eq:fAGM11k}), the fractional
FAGM(1,1,$k$) model is reduced to the fractional FAGM(1,1) model
\cite{Wu2013GreyCNSNS} with the form
\begin{eqnarray}\label{eq11:FAGM11}
\frac{{dx^{\left( r \right)} (t)}}{{dt}} + aX^{\left( r \right)} (t)
= c.
\end{eqnarray}

Setting $r=1$ in Eq. (\ref{Eq:fAGM11k}), the fractional
FAGM(1,1,$k$) model is reduced to the GM(1,1,$k$,$c$) model
\cite{Chen2014FoundationMPE} with the form
\begin{eqnarray}\label{Eq12:GM11kc}
\frac{{dx^{\left( 1 \right)} (t)}}{{dt}} + ax^{\left( 1 \right)} (t)
= bt + c.
\end{eqnarray}

Setting $r=1$, $c=0$ in Eq. (\ref{Eq:fAGM11k}), the fractional
FAGM(1,1,$k$) model is reduced to the GM(1,1,$k$) model
\cite{Cui2013AAPM} with the form
\begin{eqnarray}\label{Eq13:GM11k}
\frac{{dx^{\left( 1 \right)} (t)}}{{dt}} + ax^{\left( 1 \right)} (t)
= bt.
\end{eqnarray}

Setting $r=1$, $b=0$, $c=1$ in Eq. (\ref{Eq:fAGM11k}), the
fractional FAGM(1,1,$k$) model is reduced to the GM(1,1) model
\cite{Deng1982ControlSCL} with the form
\begin{eqnarray}\label{GM11}
\frac{{dx^{\left( 1 \right)} (t)}}{{dt}} + ax^{\left( 1 \right)} (t)
= c.
\end{eqnarray}

Thereafter, the flaw of the FAGM(1,1,$k$) model is provided. Integrating
both sides of Eq. (\ref{Eq:fAGM11k}) in the interval $[k-1, k]$,
we obtain
\begin{eqnarray}\label{Eq15:k-1tok}
\int_{k - 1}^k {dx^{\left( r \right)} (t)}  + \int_{k - 1}^k
{ax^{\left( r \right)} (t)dt}  = \int_{k - 1}^k {btdt}  + \int_{k -
1}^k c{dt}.
\end{eqnarray}

With the knowledge of $\int_{k - 1}^k {dx^{\left( r \right)} (t)}  =
x^{\left( r \right)} (k) - x^{\left( r \right)} (k - 1) = x^{\left(
{r - 1} \right)} (k)$, $\int_{k - 1}^k {tdt}  = {{\left( {2k - 1}
\right)} \mathord{\left/
 {\vphantom {{\left( {2k - 1} \right)} 2}} \right.
 \kern-\nulldelimiterspace} 2}$ and $\int_{k - 1}^k {dt = 1}$,
the exact discrete differential equation is expressed by
\begin{eqnarray}\label{Eq16:exacteqa}
x^{\left( {r - 1} \right)} (k) + a\int_{k - 1}^k {x^{\left( r
\right)} (t)dt}  = \frac{{\left( {2k - 1} \right)b}}{2} + c.
\end{eqnarray}

A comparison between Eq. (\ref{Eq16:exacteqa}) and the basic Eq.
(\ref{Eq:differential}) indicates that differences exist in
the background value $z^{\left( r \right)} (k) = 0.5\left( {x^{\left( {r
- 1} \right)} (k - 1) + x^{\left( r \right)} (k)} \right)$ and
$\int_{k - 1}^k {x^{\left( r \right)} (t)dt}$. It is highly inaccurate
to compute the integration utilising the trapezoid formula if
$x^{\left( r \right)} (t)$ is not a linear function. Thus, the basic
form and whitenisation differential equation of the
FAGM(1,1,$k$) model do not strictly match.
%
\section{Parameter optimisation of FAGM(1,1,$k$) model}
\label{sec:params}

It can easily be verified that the parameters $\phi  = \left( {a,b,c}
\right)$, derived by the least-squares estimation in Eq.
(\ref{Eq:differential}) and the parameters of the time response
function $\hat x^{\left( r \right)} (k),k = 2,3, \ldots,n$ derived
by Eq. (\ref{Eq:fAGM11k}), have different meanings. When the response
function dose not satisfy the basic equation, large errors may arise.
To match the basic Eq. (\ref{Eq:differential}) and response
function (\ref{Eq:responsefunction}), the
system parameters are optimised in this system.

Setting the optimised parameters of the grey system as $\left( {\alpha
,\beta,\gamma } \right)$ and replacing the parameters $\phi  =
\left( {a,b,c} \right)$ in Eq. (\ref{Eq:fAGM11k}), the whitening
differential equation is rewritten as
\begin{eqnarray}\label{Eq17: optimal}
\frac{{dx^{\left( r \right)} (t)}}{{dt}} + \alpha x^{\left( r
\right)} (t) = \beta t + \gamma,\ r > 0.
\end{eqnarray}

Similarly, the general solution of Eq. (\ref{Eq17: optimal}) is
given by
\begin{eqnarray}\label{Eq18:generalsolution}
x^{\left( r \right)} (t) = \left[ {x^{\left( 0 \right)} (1) -
\frac{\beta }{\alpha } + \frac{\beta }{{\alpha ^2 }} - \frac{\gamma
}{\alpha }} \right]e^{ - \alpha (t - 1)}  + \frac{\beta }{\alpha }t
- \frac{\beta }{{\alpha ^2 }} + \frac{\gamma }{\alpha }.
\end{eqnarray}

Furthermore, we have
\begin{eqnarray}\label{Eq19:xhatk}
\hat x^{\left( r \right)} (k)\! =\! \left[ {x^{\left( 0 \right)} (1)
- \frac{\beta }{\alpha }\! + \!\frac{\beta }{{\alpha ^2 }} -
\frac{\gamma }{\alpha }} \right]e^{ - \alpha (k - 1)}  + \frac{\beta
}{\alpha }k - \frac{\beta }{{\alpha ^2 }} + \frac{\gamma }{\alpha
},k = 2,3, \ldots,n.
\end{eqnarray}

Substituting Eq. (\ref{Eq19:xhatk}) into the left side of Eq.
(\ref{Eq:differential}), we obtain
\begin{eqnarray}\label{Eq20:leftvalue}
&&\hskip-3mm L(t)=\left( {x^{\left( r \right)} (k) - x^{\left( r
\right)} (k - 1)} \right)
  + \frac{a}{2}\left( {x^{\left( {r - 1} \right)} (k - 1)
  + x^{\left( r \right)} (k)} \right) \nonumber\\
&& \hskip4.5mm = \left( {1 + \frac{a}{2}} \right)x^{\left( r
\right)} (k)
  - \left( {1 - \frac{a}{2}} \right)x^{\left( {r - 1} \right)} (k - 1) \nonumber\\
&& \hskip4.5mm  = \left( {1 + \frac{a}{2}} \right)
     \left[ {\left( {x^{\left( 0 \right)} (1) - \frac{\beta }{\alpha }
      + \frac{\beta }{{\alpha ^2 }} - \frac{\gamma }{\alpha }} \right)
      e^{ - \alpha (k - 1)}  + \frac{\beta }{\alpha }k - \frac{\beta }{{\alpha ^2 }}
      + \frac{\gamma }{\alpha }} \right] \nonumber\\
&& \hskip6mm  - \left( {1 - \frac{a}{2}} \right)
    \left[ {\left( {x^{\left( 0 \right)} (1) - \frac{\beta }{\alpha }
     + \frac{\beta }{{\alpha ^2 }} - \frac{\gamma }{\alpha }} \right)
     e^{ - \alpha (k - 2)}  + \frac{\beta }{\alpha }\left( {k - 1} \right)
     - \frac{\beta }{{\alpha ^2 }} + \frac{\gamma }{\alpha }} \right] \nonumber\\
&& \hskip4.5mm  = \left[ {\left( {1 + \frac{a}{2}} \right)
   - \left( {1 - \frac{a}{2}} \right)e^\alpha  } \right]
   \left( {x^{\left( 0 \right)} (1) - \frac{\beta }{\alpha }
    + \frac{\beta }{{\alpha ^2 }} - \frac{\gamma }{\alpha }} \right)
    e^{ - \alpha (k - 1)} \nonumber \\
&& \hskip6mm  + \frac{\beta }{\alpha }ak + \left( {1 - \frac{a}{2}}
\right)
   \frac{\beta }{\alpha } - \left( {\frac{\beta }{{\alpha ^2 }}
   - \frac{\gamma }{\alpha }} \right)a.
\end{eqnarray}

Owing to the left side $L(t)$ and right side $R(t)$ equivalence,
namely $L(t)-R(t)=0$, it is implied that
\begin{eqnarray}
&&\left( {1 + \frac{a}{2}} \right) - \left( {1 - \frac{a}{2}}
\right)e^\alpha   = 0, \label{Eq21:first} \\
&&\frac{\beta }{\alpha }a = b, \label{Eq22:second} \\
&&\left( {1 - \frac{a}{2}} \right)\frac{\beta }{\alpha } - \left(
{\frac{\beta }{{\alpha ^2 }} - \frac{\gamma }{\alpha }} \right)a = c
- \frac{b}{2}. \label{Eq23:third}
\end{eqnarray}

It follows from Eqs. (\ref{Eq21:first}) to (\ref{Eq23:third}) that
\begin{eqnarray}
&&\alpha  = \ln \frac{{2 + a}}{{2 - a}},\label{Eq24:one} \\
&&\beta  = \frac{b}{a}\ln \frac{{2 + a}}{{2 - a}},\label{Eq25:two} \\
&&\gamma  = \frac{{\alpha c}}{a} - \frac{{\alpha b}}{{2a}} +
\frac{\beta }{\alpha } + \frac{\beta }{2} - \frac{\beta }{a}.
\label{Eq26:three}
\end{eqnarray}

Thus, the optimised parameters $\left( {\alpha,\beta,\gamma }
\right)$ are obtained by Eqs. (\ref{Eq24:one}) to (\ref{Eq26:three}),
and they also indicate that the parameters $\left( {a,b,c} \right)$ derived
by the least-squares estimation satisfy the relationship in Eqs.
(\ref{Eq24:one}) to (\ref{Eq26:three}). In this paper, the
FAGM(1,1,$k$) model with optimised parameters is referred to as the
FAGMO(1,1,$k$) model.

\begin{theorem}\label{Th4:eqals}
Assuming that the original data $\left( {x^{(0)} (1),x^{(0)} (2),
\ldots,x^{(0)} (n)} \right)$ satisfy Eq. (\ref{Eq19:xhatk})
with the given parameters $\left( {\hat \alpha,\hat \beta,\hat \gamma
} \right)$, the parameters $\left( {\alpha,\beta,\gamma }
\right)$ of FAGMO(1,1,$k$) obtained by Eqs. (\ref{Eq:esparams}) and
(\ref{Eq24:one}) to (\ref{Eq26:three}) satisfy the relationship $\hat
\alpha  = \alpha,\hat \beta  = \beta,\hat \gamma  = \gamma$, and the
predicted values $\left( {\hat x^{(0)} (1),\hat x^{(0)} (2), \ldots
,\hat x^{(0)} (n)} \right)$ are equal to the given data $\left( {x^{(0)}
(1),x^{(0)} (2), \ldots,x^{(0)} (n)} \right)$.
\end{theorem}

\begin{proof}
Substituting the original data $x^{(0)} (k),k = 1,2, \ldots,n$ into
Eq. (\ref{Eq:esparams}), the parameters $\left( {\hat \alpha,\hat \beta,\hat \gamma} \right)$ can be
derived. The parameters $\left( {\alpha,\beta,\gamma } \right)$ of
FAGMO(1,1,$k$) can be obtained from Eqs.
(\ref{Eq24:one}) to (\ref{Eq26:three}). Obviously, $\hat \alpha  =
\alpha,\hat \beta  = \beta,\hat \gamma  = \gamma$. Therefore, the
predicted values $\hat x^{(0)} (k),k = 1,2, \ldots,n$ of the
FAGMO(1,1,$k$) are equal to the given data $x^{(0)} (k),k = 1,2,
\ldots,n$.
\end{proof}

Theorem \ref{Th4:eqals} demonstrates that the FAGMO(1,1,$k$) model is
accurate for predicting arbitrary sequences that can be modelled by Eq.
(\ref{Eq19:xhatk}), while the FAGM(1,1,$k$) model cannot describe the
sequences accurately owing to there always being a non-zero difference
between the real parameters and the parameters $\left( {a,b,c}
\right)$.

\begin{theorem}\label{Th5:approximately}
The optimised parameters $\left( {\alpha,\beta,\gamma } \right)$
are approximately equivalent to the parameters $\left( {a,b,c}
\right)$ when the value of $\left| a \right|$ is very small; that
is,
\begin{eqnarray}\label{Eq27:approximately}
\alpha  \approx a,\ \beta  \approx b,\ \gamma  \approx c.
\end{eqnarray}
\end{theorem}

\begin{proof}
We first consider the difference between parameter $\alpha$ and $a$,
which is
\begin{eqnarray}\label{Eq28:varepsilon1}
\varepsilon _1 \left( a \right) = \alpha  - a = \ln \frac{{2 +
a}}{{2 - a}} - a.
\end{eqnarray}

It is known that $\left. {\varepsilon _1 \left( a \right)}
\right|_{a = 0}  = 0$ and the first-order derivative is
\begin{eqnarray}\label{Eq29:derivation}
\frac{{d\varepsilon _1 \left( a \right)}}{{da}} = \frac{{d\left(
{\ln \frac{{2 + a}}{{2 - a}} - a} \right)}}{{da}} = \frac{4}{{4 -
a^2 }} - 1 = \frac{{a^2 }}{{4 - a^2 }}.
\end{eqnarray}

When $\left| a \right| < 2$, the derivative of $\varepsilon _1
\left( a \right)$ is positive, which indicates that the function
$\varepsilon _1 \left( a \right)$ is a monotonically increasing
function in the interval [-2,2]. Thus, the value $\varepsilon _1
\left( a \right)$ approaches zero as $\left| a \right|$
decreases. Therefore, $\alpha  \approx a$ when $\left| a \right|$ is
very small.

Secondly, the difference between $\beta$ and $b$ is expressed as
\begin{eqnarray}\label{Eq30:betab}
\varepsilon _2 \left( a \right) = \beta  - b = \frac{b}{a}\ln
\frac{{2 + a}}{{2 - a}} - b = \frac{b}{a}\varepsilon _1 \left( a
\right).
\end{eqnarray}

Owing to $\mathop {\lim }\limits_{a \to 0} \frac{1}{a}\ln \frac{{2 +
a}}{{2 - a}} = \mathop {\lim }\limits_{a \to 0} \frac{{2 - a}}{{2 +
a}}\frac{{2 - a + 2 + a}}{{\left( {2 - a} \right)^2 }} = \mathop
{\lim }\limits_{a \to 0} \frac{4}{{4 - a^2 }} = 1$, we know that
$\varepsilon _2 \left( a \right) \to 0$ when $a \to 0$.

The first-order derivative of $\varepsilon _2 \left( a \right)$ is
\begin{eqnarray}\label{Eq31:derivativevarepsilon2}
\frac{{d\varepsilon _2 \left( a \right)}}{{da}} = \frac{b}{{a^2
}}\left( {\frac{{4a}}{{4 - a^2 }} - \ln \frac{{2 + a}}{{2 - a}}}
\right),
\end{eqnarray}
which is also positive when $\left| a \right| < 2$. Thus,
$\varepsilon _2 \left( a \right)$ decreases when the value of
$\left| a \right|$ decreases and $\beta  \approx b$ when
$\left| a \right|$ is very small.

Thirdly, the difference between $\gamma$ and $c$ is expressed as
\begin{eqnarray}\label{Eq32:differenceevarepsilon2}
\varepsilon _3 (a) = \gamma  - c = \frac{{\alpha c}}{a} -
\frac{{\alpha b}}{{2a}} + \frac{\beta }{\alpha } + \frac{\beta }{2}
- \frac{\beta }{a} - c = \frac{{\varepsilon _1 \left( a
\right)}}{a}\left( {c - \frac{b}{a}} \right).
\end{eqnarray}

It follows from $\alpha  \approx a$ and $\beta  \approx b$ that
$\gamma  \approx c$ when $\left| a \right|$ is very small.
\end{proof}

From Theorem \ref{Th5:approximately}, we know that the differences
between the parameters $\left( {\alpha,\beta,\gamma } \right)$ and
$\left( {a,b,c} \right)$ are decrease along with smaller
$\left| a \right|$. Table \ref{table1:values} provides the values
of $\varepsilon _1 \left( a \right)$ and ${{\varepsilon _1 \left( a
\right)} \mathord{\left/
 {\vphantom {{\varepsilon _1 \left( a \right)} a}} \right.
 \kern-\nulldelimiterspace} a}$ under different values of $\left| a
 \right|$.

\begin{table}[!htbp]
\caption{Values of $\varepsilon_1\left( a \right)$ and
$\varepsilon_1 \left( a \right)/a$ under different values of $\left|
a \right|$}
 \label{table1:values}
\centering \scriptsize\centerline{\tabcolsep=5pt
\begin{tabular}{lccccccccc}
\hline
 $\left|a \right|$                  &0.1 &0.2 &0.3 &0.5 &0.7 &1.0 &1.3 &1.6 &1.9  \\
 \hline
 $\varepsilon_1\left( a \right)$    &0.0001  &0.0007  &0.0023  &0.0108  &0.0309  &0.0986  &0.2506 &0.5972 &1.7636  \\
$\varepsilon_1 \left( a \right)/a$  &0.0008  &0.0034  &0.0076  &0.0217  &0.0441  &0.0986  &0.1928 &0.3733 &0.9282\\
 \hline
\end{tabular}}
\end{table}
%
\section{Modelling evaluation criteria and detailed modelling steps}
\label{sec:criteria}

To evaluate forecasting accuracy of the FAGMO(1,1,$k$)
model, the root mean squared percentage error (RMSPE)
is applied to the prior-sample period (RMSPEPR) and post-sample period
(RMSPEPO). In general, the RMSPEPR, RMSPEPO and RMSPE are
defined as
\begin{eqnarray}
&&{\rm RMSPEPR} = \sqrt {\frac{1}{\nu }\sum\limits_{k = 1}^\nu
{\left( {\frac{{\hat x_1^{\left( 0 \right)} (k) - x_1^{\left( 0
\right)}
(k)}}{{x_1^{\left( 0 \right)} (k)}}} \right)^2 } }  \times 100\%,\label{Eq33:RMSPEPR}\\
&&{\rm RMSPEPO} = \sqrt {\frac{1}{{n - \nu }}\sum\limits_{k = \nu  +
1}^n {\left( {\frac{{\hat x_1^{\left( 0 \right)} (k) - x_1^{\left( 0
\right)} (k)}}{{x_1^{\left( 0 \right)} (k)}}} \right)^2 } } \times
100\%,\label{Eq34:RMSPEPO}\\
&&{\rm RMSPE} = \sqrt {\frac{1}{n}\sum\limits_{k = 1}^n {\left(
{\frac{{\hat x_1^{\left( 0 \right)} (k) - x_1^{\left( 0 \right)}
(k)}}{{x_1^{\left( 0 \right)} (k)}}} \right)^2 } }  \times
100\%,\label{Eq35:RMSPE}
\end{eqnarray}
where $\nu$ is the number of samples used to construct the model
and $n$ is the total number of samples.

The index of agreement of the forecasting results is defined as
\begin{eqnarray}
{\rm IA} = 1 - \frac{{\sum\nolimits_{k = 1}^n {\left( {\hat
x^{\left( 0 \right)} (k) - x^{\left( 0 \right)} (k)} \right)^2 }
}}{{\sum\nolimits_{k = 1}^n {\left( {\left| {\hat x^{\left( 0
\right)} (k) - \overline x} \right| + \left| { x^{\left( 0 \right)}
(k) - \overline x} \right|} \right)^2 } }},\label{Eq36:IA}
\end{eqnarray}
which is also a useful performance measure for sensitivity to
differences in the observed and predicted data, where $\overline x$ is
the average sample value.

The average forecasting error (AE) and the mean absolute forecasting
error (MAE) are
\begin{eqnarray}
&&{\rm AE} = \frac{1}{n}\sum\limits_{k = 1}^n {\left( {\hat
x^{\left( 0 \right)} (k) - x^{\left( 0 \right)} (k)} \right)},\label{Eq37:AE}\\
&&{\rm MAE} = \frac{1}{n}\sum\limits_{k = 1}^n {\left| {\hat
x^{\left( 0 \right)} (k) - x^{\left( 0 \right)} (k)}
\right|},\label{Eq38:MAE}
\end{eqnarray}
where AE reflects the positive and negative errors between the predicted and
observed values, while MAE is applied for estimating the change in the
forecasting model.

The detailed modelling steps of the fractional FAGMO(1,1,$k$) are
provided below.

{\bf Step 1:} Determine the original data series $x^{\left( 0
\right)} (i),i = 1,2, \ldots,n$, and $r$-AGO series $X^{\left(
r \right)}  = X^{\left( 0 \right)} A^r$.

{\bf Step 2:} Calculate the matrices $B$ and $Y$ to determine $\left(
{a,b,c} \right)$ using Eq. (\ref{Eq:esparams}).

{\bf Step 3:} Compute the parameters $\left( {\alpha,\beta,\gamma
} \right)$ by employing Eqs. (\ref{Eq24:one}) to (\ref{Eq26:three}).

{\bf Step 4:} Substitute the values of $x^{(0)} (1)$ and $\left(
{\alpha,\beta,\gamma } \right)$  into Eq. (\ref{Eq19:xhatk}) to
compute the predicted values $\hat X^{(r)}$.

{\bf Step 5:} Apply the $r$-IAGO matrix to obtain the restored
values $\hat X^{(0)}  = \hat X^{(r)} D^r$.

%
\section{Validation of FAGMO(1,1,$k$) model}
\label{sec:validation}

This section provides numerical examples to validate the accuracy of
the FAGMO(1,1,$k$) model compared to the FAGM(1,1,$k$) model and others.

\subsection{Validation of FAGMO(1,1,$k$) and FAGM(1,1,$k$) models}
 \label{subsec6.1:validation of}

This subsection presents a numerical example to validate the
accuracy of the FAGMO(1,1,$k$) and FAGM(1,1,$k$) models.
The values $r$ and $\alpha$ are provided in the interval [0.01, 2] and
[-1.99, 1.99], respectively. The initial point $x^{\left( r \right)}
(1)$ is randomly generated in the interval [1, 2] by the uniform
distribution, while the parameters $\beta$ and $\gamma$ are randomly generated in the intervals [0, 5] and [0, 100], respectively,
by the uniform distribution. The other $x^{\left( r \right)} \left(
i \right)\left( {i > 1} \right)$ are generated with the aid of Eq.
(\ref{Eq19:xhatk}). All data used for the example are explained in Fig.
\ref{fig5:Validation }.
\begin{figure}[!htbp]
\centering\centerline{\includegraphics[height=5.2cm,width=10cm]{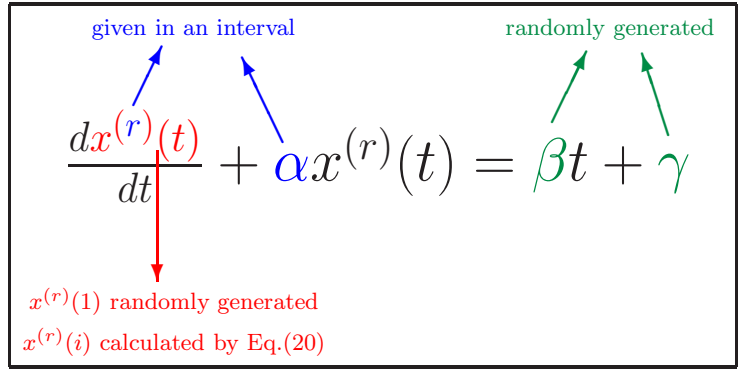}}
\caption{Diagram of data for validation}
 \label{fig5:Validation }
\end{figure}

We define the notation in the following analysis
\begin{eqnarray}\label{Eq39:varepsilonparams}
\varepsilon _{\rm {params}}  = \left( {p - \alpha } \right)^2  +
\left( {l - \beta } \right)^2  + \left( {q - \gamma } \right)^2,
\end{eqnarray}
where $\left( {\alpha,\beta,\gamma } \right)$ are the provided parameters
of Eq. (\ref{Eq19:xhatk}) and $\left( {p,l,q} \right)$ are the
estimated parameters of the FAGM(1,1,$k$) or
FAGMO(1,1,$k$) model.

When applying the above parameters, the graphs are displayed in Figs.
\ref{fig6:varepsilonparams} and \ref{fig7:RMSPE}. We observe from
Fig. \ref{fig6:varepsilonparams} that the maximum $\varepsilon
_{\rm {params}}$ of FAGMO(1,1,$k$) and FAGM(1,1,$k$) are
$5.4228\times 10^{-5}$ and 489.9434, respectively, where the
magnitude is approximately 9034932. Furthermore, the $\varepsilon _{\rm
{params}}$ of the FAGM(1,1,$k$) model is very small when $\alpha$ is
near zero, which is coincident with Theorem
\ref{Th5:approximately}. From Fig. \ref{fig7:RMSPE}, the maximum
RMSPEs of FAGMO(1,1,$k$) and FAGM(1,1,$k$) are 0.0103\% and
814.3864\%, respectively, where the magnitude is approximately 79100.

It is known that the parameters $\beta$, $\gamma$, and initial points
$x^{\left( r \right)} (1)$ are all randomly generated, which implies
that the values of parameters $\beta$, $\gamma$ and $x^{\left( r
\right)} (1)$ have no influence on the output series. Here, the
values $r$ and $\alpha$ are the most important factors affecting
the accuracy of the grey models.

\begin{figure}[!htbp]
\centering
\begin{minipage}[c]{0.5\textwidth}
\centering
\includegraphics[height=5.5cm,width=6.2cm]{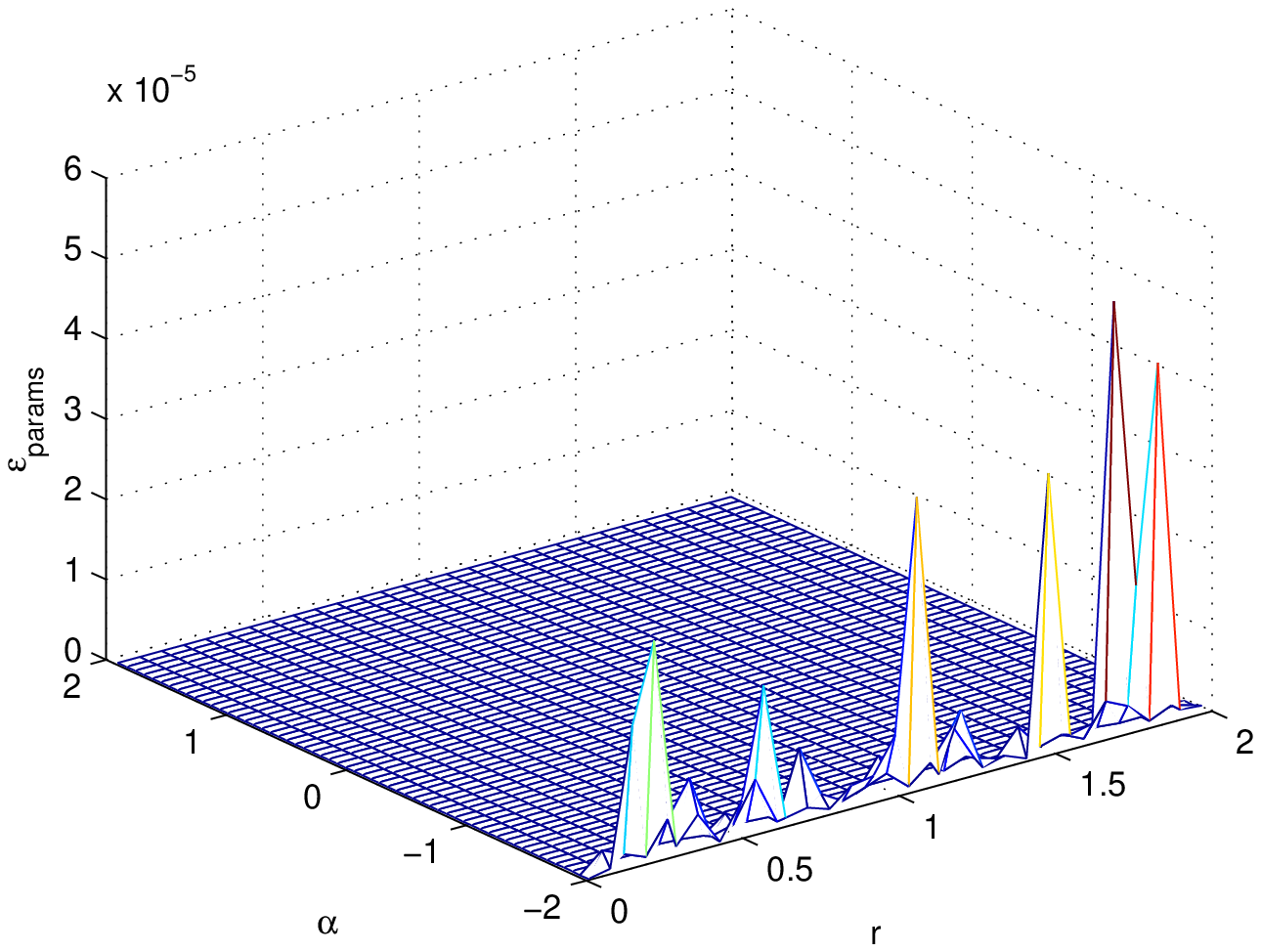}
\end{minipage}%
\begin{minipage}[c]{0.5\textwidth}
\centering
\includegraphics[height=5.5cm,width=6.2cm]{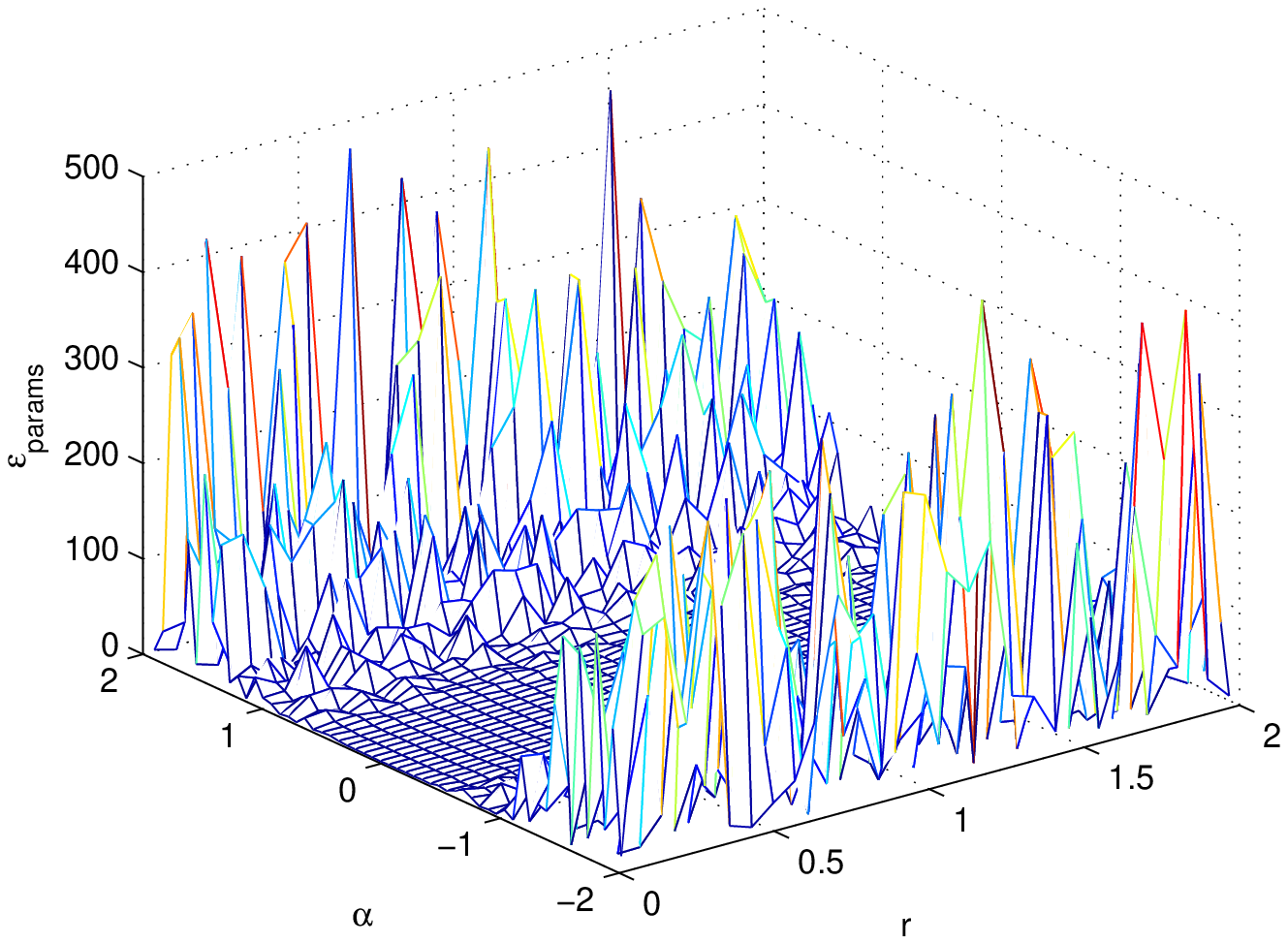}
\end{minipage}
\caption{Values of $\varepsilon _{\rm {params}}$ of
FAGMO(1,1,$k$) (left) and FAGM(1,1,$k$) (right) models}
 \label{fig6:varepsilonparams}
\end{figure}  \vskip-2mm

\begin{figure}[!htbp]
\centering
\begin{minipage}[c]{0.5\textwidth}
\centering
\includegraphics[height=5.5cm,width=6.2cm]{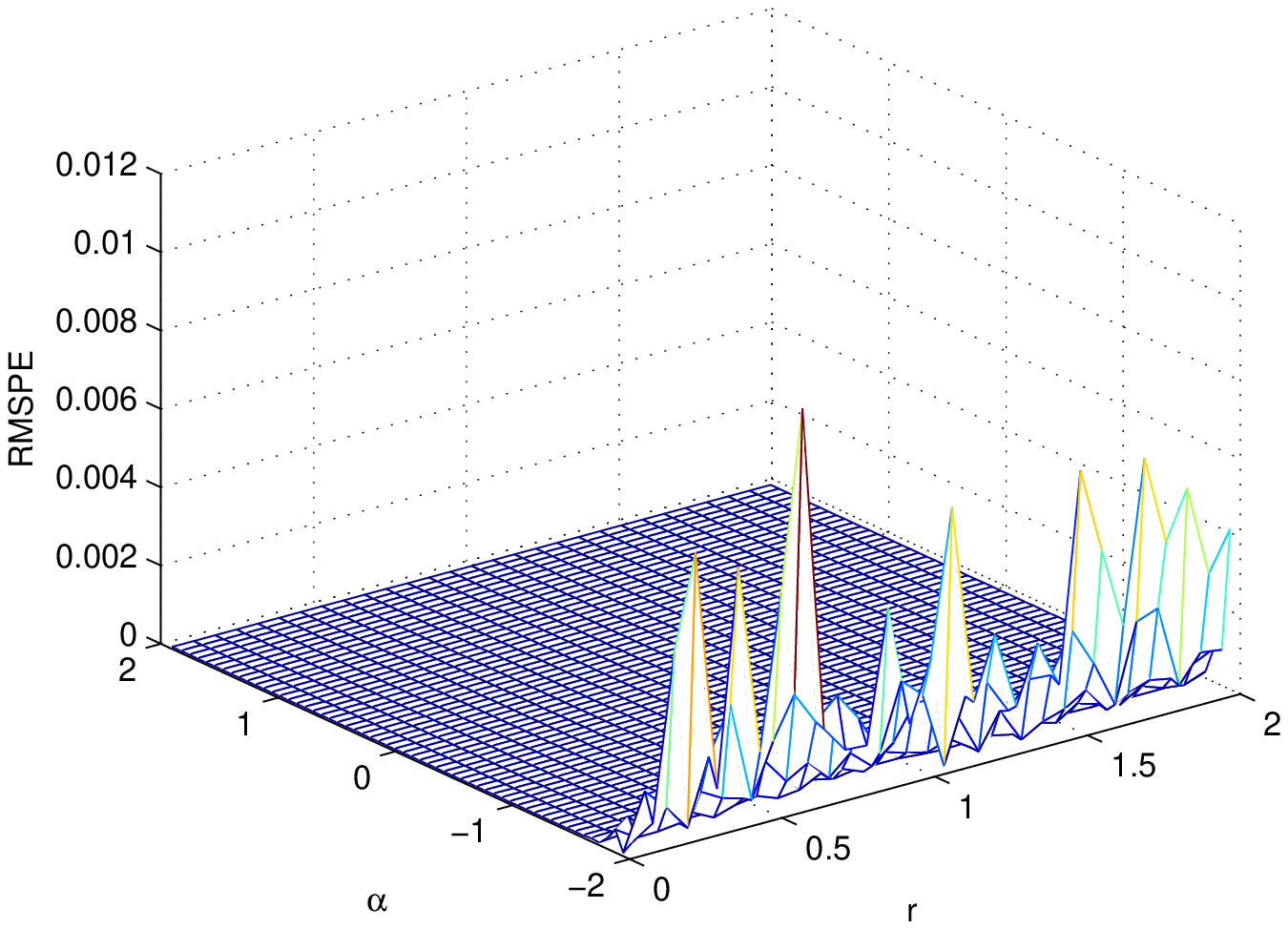}
\end{minipage}%
\begin{minipage}[c]{0.5\textwidth}
\centering
\includegraphics[height=5.5cm,width=6.2cm]{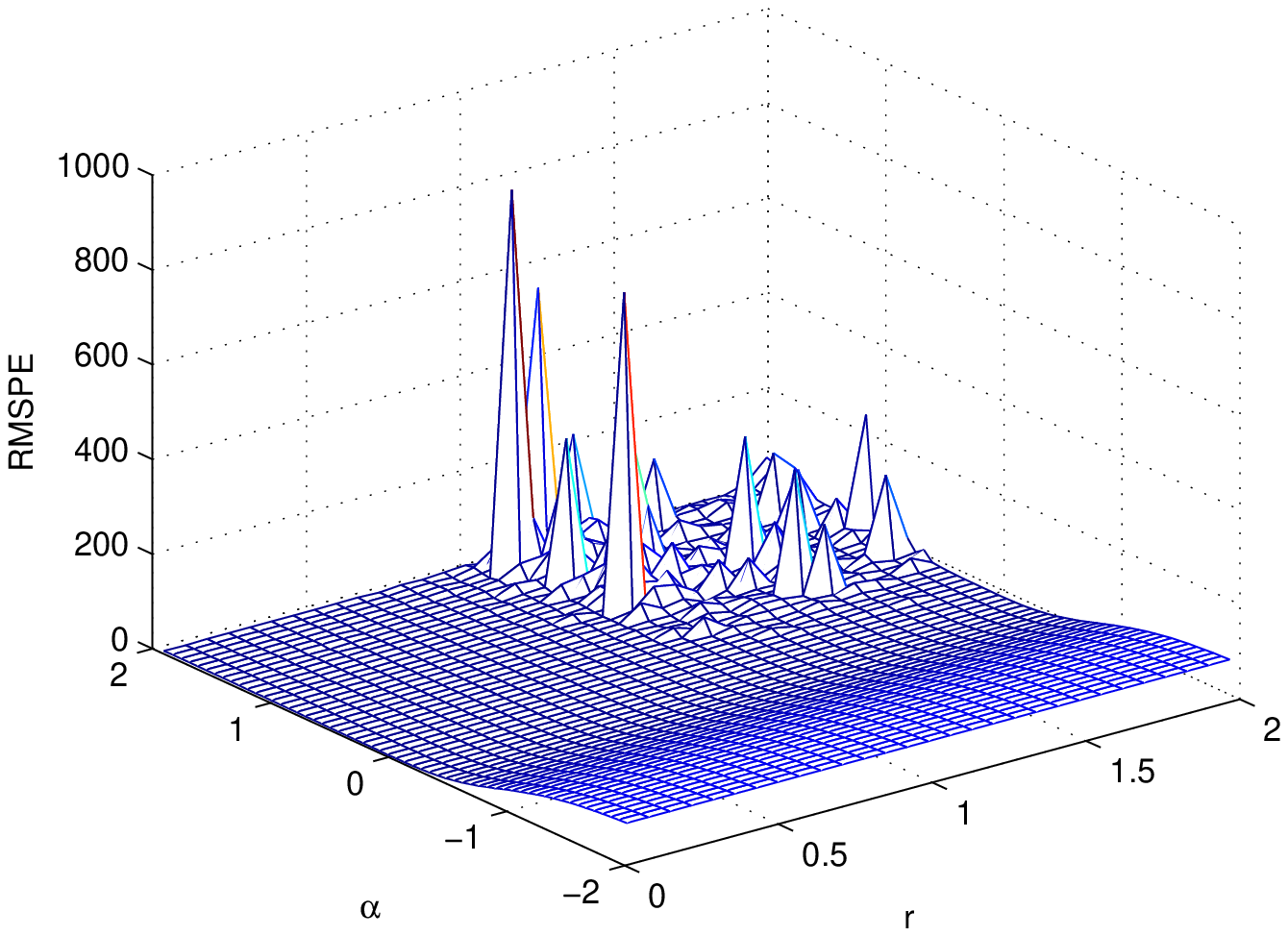}
\end{minipage}
\caption{Values of RMSPE of FAGMO(1,1,$k$) (left) and
FAGM(1,1,$k$) (right) models}
 \label{fig7:RMSPE}
\end{figure}
%
\subsection{Validation of FAGMO(1,1,$k$) model and other grey models}
 \label{subsec6.2:Validationother}

This subsection further demonstrates the advantage of the
FAGMO(1,1,$k$) model using two real cases.

{\bf Case 1:} (Predicting cumulative oil field production). We
consider an example from the paper \cite{Ma2016PredictingJCTN} that
provides sample data. The data from 1999 to 2009 are applied to construct
the grey model, while the data from 2010 to 2012 are used for prediction. The
values are listed in Table \ref{table2:validation}, indicating
that the FAGMO(1,1,$k$) model outperforms the other models in this
case.

\begin{table}[!htbp]
\caption{Results of ENGM, FAGM(1,1), FAGM(1,1,$k$) and FAGMO(1,1,$k$) models}
 \label{table2:validation}
 \centering
 \scriptsize
 \begin{tabular}{cccccccccccccc}
 \hline
Year    &Data        &  ENGM      & FAGM(1,1)  & FAGM(1,1,$k$) &FAGMO(1,1,$k$)\\
\hline
        &            & $r=1$      &$r=0.1106$  &$r=0.4073$     &$r=0.4052$\\
1999    &73.8217     &73.8217     &73.8217     &73.8217        &73.8217\\
2000    &136.8817    &138.4900    &138.1621    &137.1758       &136.4573\\
2001    &195.0590    &195.4541    &195.5377    &196.1598       &195.7633\\
2002    &247.8547    &247.9776    &247.7638    &249.3183       &249.1781\\
2003    &297.0902    &296.4067    &295.7629    &297.2895       &297.2750\\
2004    &342.6394    &341.0604    &340.1238    &341.0008       &341.0322\\
2005    &382.4312    &382.2332    &381.2700    &381.2882       &381.3320\\
2006    &420.0399    &420.1964    &419.5291    &418.8204       &418.8699\\
2007    &454.0430    &455.2001    &455.1670    &454.1099       &454.1712\\
2008    &485.1171    &487.4752    &488.4068    &487.5452       &487.6290\\
2009    &519.8508    &517.2342    &519.4402    &519.4217       &519.5393\\[5pt]
2010    &552.6569    &544.6734    &548.4350    &549.9665       &550.1281\\
2011    &581.6092    &569.9736    &575.5400    &579.3572       &579.5714\\
2012    &608.1863    &593.3015    &600.8887    &607.7346       &608.0086\\[5pt]
\multicolumn{2}{l}{RMSPEPR}&0.4521\%    &0.4582\%    &0.3539\%       &{\bf 0.3259\%}\\
\multicolumn{2}{l}{RMSPEPO}&2.0066\%    &1.0185\%    &0.3617\%       &{\bf 0.3332\%}\\
\hline
\end{tabular}
\end{table}

{\bf Case 2:}  (Predicting foundation settlement close
neighbouring Yangtze River). We consider an example from the
paper \cite{Chen2014FoundationMPE}, which provides sample data to
construct the grey model. The values are presented in Table
\ref{table3:case2}, indicating that the FAGMO(1,1,$k$) model
outperforms the other models in this case.

\begin{table}[!htbp]
\caption{Results of ONGM(1,1,$k$,$c$), FAGM(1,1), FAGM(1,1,$k$) and FAGMO(1,1,$k$) models}
 \label{table3:case2}
 \centering
 \scriptsize
 \begin{tabular}{cccccccccccccc}
 \hline
Day &Data   &ONGM(1,1,$k$,$c$) &FAGM(1,1) &FAGM(1,1,$k$) & FAGMO(1,1,$k$)\\
 \hline
             &            &$r=1$      &$r=0.0065$   &$r=0.2266$   &$r=0.2295$\\
10  &23.36   &23.3600     &23.3600     &23.3600     &23.3600\\
20  &43.19   &42.1779     &43.3517     &43.0586     &43.0644\\
30  &58.73   &59.2549     &59.4403     &58.8205     &58.8124\\
40  &70.87   &72.8374     &72.4009     &71.9763     &71.9545\\
50  &83.71   &83.6405     &82.8451     &83.0247     &82.9932\\
60  &92.91   &92.2330     &91.2620     &92.2158     &92.1789\\
70  &99.73   &99.0672     &98.0442     &99.6885     &99.6491\\
80  &105.08  &104.5030    &103.5079    &105.5215    &105.4805\\
90  &109.73  &108.8264    &107.9079    &109.7568    &109.7127\\
100 &112.19  &112.2652    &111.4497    &112.4117    &112.3598\\
110 &113.45  &115.0002    &114.2991    &113.4857    &113.4181\\[5pt]
\multicolumn{2}{l}{RMSPE} &1.2730\%     &1.3257\%     &0.6030\%     &{\bf0.6011\%}\\
\hline
\end{tabular}
\end{table}

%
\newpage
\section{ Applications}
 \label{sec:appli}

In this section, the FAGMO(1,1,$k$) model is applied to forecast the
nuclear energy consumption of China. The computational results of the
FAGMO(1,1,$k$) model are compared to the ENGM
\cite{Ma2016PredictingJCTN}, ONGM(1,1,$k$)
\cite{Chen2014FoundationMPE}, FAGM(1,1) \cite{Wu2013GreyCNSNS} and
FAGM(1,1,$k$) models.

\subsection{Raw data}
\label{subsec:raw data}

Raw data of the nuclear energy consumption of China were collected
from the report of the {\it BP Statistical Review of World Energy
2018}. The first 10 samples belonging to the $11^{\rm th}$ and the
$12^{\rm th}$ Five-Year Plans are applied to construct the
prediction model, while the remaining samples of the $13^{\rm th}$
Five-Year Plan are used to validate and compare the forecasting
results (see Table \ref{table:raw-data}).
\begin{table}[!htbp]
\caption{Raw data of nuclear energy consumption of China, Mtoe}
 \label{table:raw-data}
 \centering
 \scriptsize
 \begin{tabular}{cccccccccccccc}
 \hline
Year &Data &&&Year  &Data        &&&Year   &Data    \\
\hline
2006 &12.4 &&&2011  &19.5        &&&2016     &48.2       \\
2007 &14.1 &&&2012  &22.0        &&&2017     &56.2    \\
2008 &15.5 &&&2013  &25.3         \\
2009 &15.9 &&&2014  &30.0         \\
2010 &16.7 &&&2015  &38.6         \\
 \hline
\end{tabular}
\end{table}
\subsection{Simulation and prediction results}
\label{subsec:sim-pre-results}

The simulation and prediction results are listed in Table
\ref{table4:Predictionresults} and Fig. \ref{fig8:Comparison}, while
the errors are listed in Table \ref{table5:Predictionerror} and Fig.
\ref{fig9:error}.

The nuclear energy consumption of China from 2016 to 2017 is
predicted according to the established grey models. It can be observed
in Table \ref{table4:Predictionresults} and Fig.
\ref{fig8:Comparison} that five grey models, namely ENGM,
ONGM(1,1,$k$), FAGM(1,1), FAGM(1,1,$k$) and FAGMO(1,1,$k$),
successfully identify the trend of China's nuclear energy
consumption. However, these grey models differ from one another in terms of the
prediction values from 2016 to 2020. From Fig.
\ref{fig8:Comparison}, China's nuclear energy consumption is
overestimated by the ENGM, ONGM(1,1,$k$) and FAGM(1,1,$k$) models, and
underestimated by the FAGM(1,1) model. The values predicted by
FAGMO(1,1,$k$) are substantially closer to the raw data than those predicted by the other
models.

We can observe from Table \ref{table5:Predictionerror} and Fig.
\ref{fig9:error} that the RMSPEPR, RMSPEPO and RMSPE of
FAGMO(1,1,$k$) are 3.1409\%, 4.1502\% and 3.3304\%, respectively.
The RMSPEPR, RMSPEPO and RMSPE of ENGM are as high as 8.3788\%,
30.3663\% and 14.5667\%, those of ONGM(1,1,$k$) are 2.0494\%,
12.0510\% and 5.2635\%, those of FAGM(1,1) are 4.8680\%, 11.7968\%
and 6.5529\%, and those of FAGM(1,1,$k$) are 2.3299\%, 6.3828\% and
3.3636\%, respectively. The IA, AE and MAE of FAGMO(1,1,$k$) are
0.9985, 0.2526 and 0.7513, those of ENGM are 0.9538, 4.0536 and
4.0536, those of ONGM(1,1,$k$) are 0.9911, 1.0225 and 1.1896, those
of FAGM(1,1) are 0.9887, -1.1818 and 1.7105, and those of
FAGM(1,1,$k$) are 0.9971, 0.2736 and 0.8043, respectively.
The computational results indicate that the FAGMO(1,1,$k$) model
outperforms ENGM, ONGM(1,1,$k$), FAGM(1,1) and FAGM(1,1,$k$),
while ENGM exhibits the most inferior performance.

\begin{table}[!htbp]
\caption{Simulation and prediction results of nuclear energy
consumption by grey models}
 \label{table4:Predictionresults}
 \centering
 \scriptsize
 \begin{tabular}{cccccccccccccc}
 \hline
Year &Data &ENGM    &ONGM(1,1,$k$)  &FAGM(1,1)   &FAGM(1,1,$k$)   &FAGMO(1,1,$k$) \\
     &     &$r=1$   &$r=1$          &$r=1.4127$  &$r=1.0593$      &$r=1.1595$ \\
\hline
2006 &12.4 &12.4000  &12.4000        &12.4000     &12.4000     &12.4000  \\
2007 &14.1 &14.9788  &14.4788        &15.0242     &14.7054     &15.0891\\
2008 &15.5 &15.6744  &15.1057        &13.9808     &15.0121     &14.8608\\
2009 &15.9 &16.6846  &16.0032        &15.0566     &15.8012     &15.5886\\
2010 &16.7 &18.1520  &17.2884        &16.9219     &17.0700     &16.9760\\
2011 &19.5 &20.2831  &19.1286        &19.3953     &18.9344     &19.0534\\
2012 &22.0 &23.3785  &21.7635        &22.4687     &21.5861     &21.9432\\
2013 &25.3 &27.8741  &25.5363        &26.1951     &25.3029     &25.8633\\
2014 &30.0 &34.4036  &30.9383        &30.6625     &30.4740     &31.1013\\
2015 &38.6 &43.8871  &38.6732        &35.9872     &37.6390     &38.0473\\[5pt]
2016 &48.2 &57.6610  &49.7483        &42.3129     &47.5433     &47.2178\\
2017 &56.2 &77.6662  &65.6063        &49.8133     &61.2149     &59.2933\\
2018 & --  &106.7219 &88.3125        &58.6959     &80.0704     &75.1679\\
2019 & --  &148.9226 &120.8244       &69.2074     &106.0614    &96.0147\\
2020 & --  &210.2150 &167.3766       &81.6403     &141.8758    &123.3723\\
 \hline
\end{tabular}
\end{table}

\begin{figure}[!htbp]
\centering\centerline{\includegraphics[height=8cm,width=12cm]{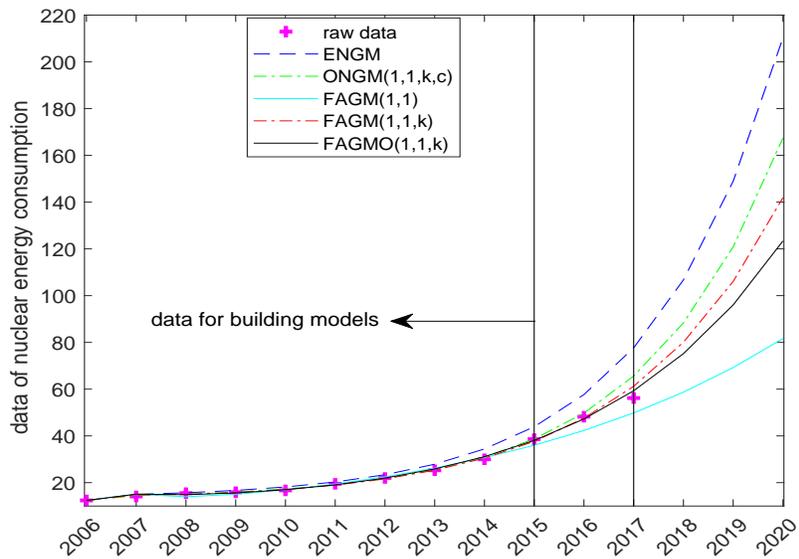}}
\caption{Comparison among five grey models for nuclear energy
consumption}
 \label{fig8:Comparison}
\end{figure}

\begin{table}[!htbp]
\caption{Relative error values of nuclear energy consumption by five
grey models}
 \label{table5:Predictionerror}
 \centering
 \scriptsize
 \begin{tabular}{cccccccccccccc}
 \hline
Year    &ENGM    &ONGM(1,1,$k$)  &FAGM(1,1)   &FAGM(1,1,$k$)  &FAGMO(1,1,$k$)\\
\hline
2006    &0       &0              &0           &0              &0        \\
2007    &0.0623  &0.0269         &0.0655      &0.0429         &0.0701   \\
2008    &0.0112  &0.0254         &0.0980      &0.0315         &0.0412   \\
2009    &0.0493  &0.0065         &0.0530      &0.0062         &0.0196   \\
2010    &0.0869  &0.0352         &0.0133      &0.0222         &0.0165   \\
2011    &0.0402  &0.0190         &0.0054      &0.0290         &0.0231   \\
2012    &0.0627  &0.0108         &0.0213      &0.0188         &0.0026   \\
2013    &0.1017  &0.0093         &0.0354      &0.0001         &0.0222   \\
2014    &0.1468  &0.0313         &0.0221      &0.0158         &0.0367   \\
2015    &0.1370  &0.0019         &0.0677      &0.0249         &0.0143   \\ [5pt]
2016    &0.1963  &0.0321         &0.1221      &0.0136         &0.0204   \\
2017    &0.3820  &0.1674         &0.1136      &0.0892         &0.0550   \\ [5pt]
RMSPEPR &8.3788\%&2.0494\%       &4.8680\%    &{\bf2.3299}\%  &3.1409\%   \\
RMSPEPO &30.3663\%&12.0510\%     &11.7968\%   &6.3828\%       &{\bf4.1502\%}   \\
RMSPE   &14.5667\%&5.2635\%      &6.5529\%    &3.3636\%       &{\bf3.3304\%}   \\
IA      &0.9538  &0.9911         &0.9887      &0.9971         &{\bf0.9985}   \\
AE      &4.0536  &1.0225         &-1.1818     &0.2736         &{\bf0.2526}   \\
MAE     &4.0536  &1.1896         &1.7105      &0.8043         &{\bf0.7513}   \\
\hline
\end{tabular}
\end{table}

\begin{figure}[!htbp]
\centering
\begin{minipage}[c]{0.5\textwidth}
\centering
\includegraphics[height=5.8cm,width=6.2cm]{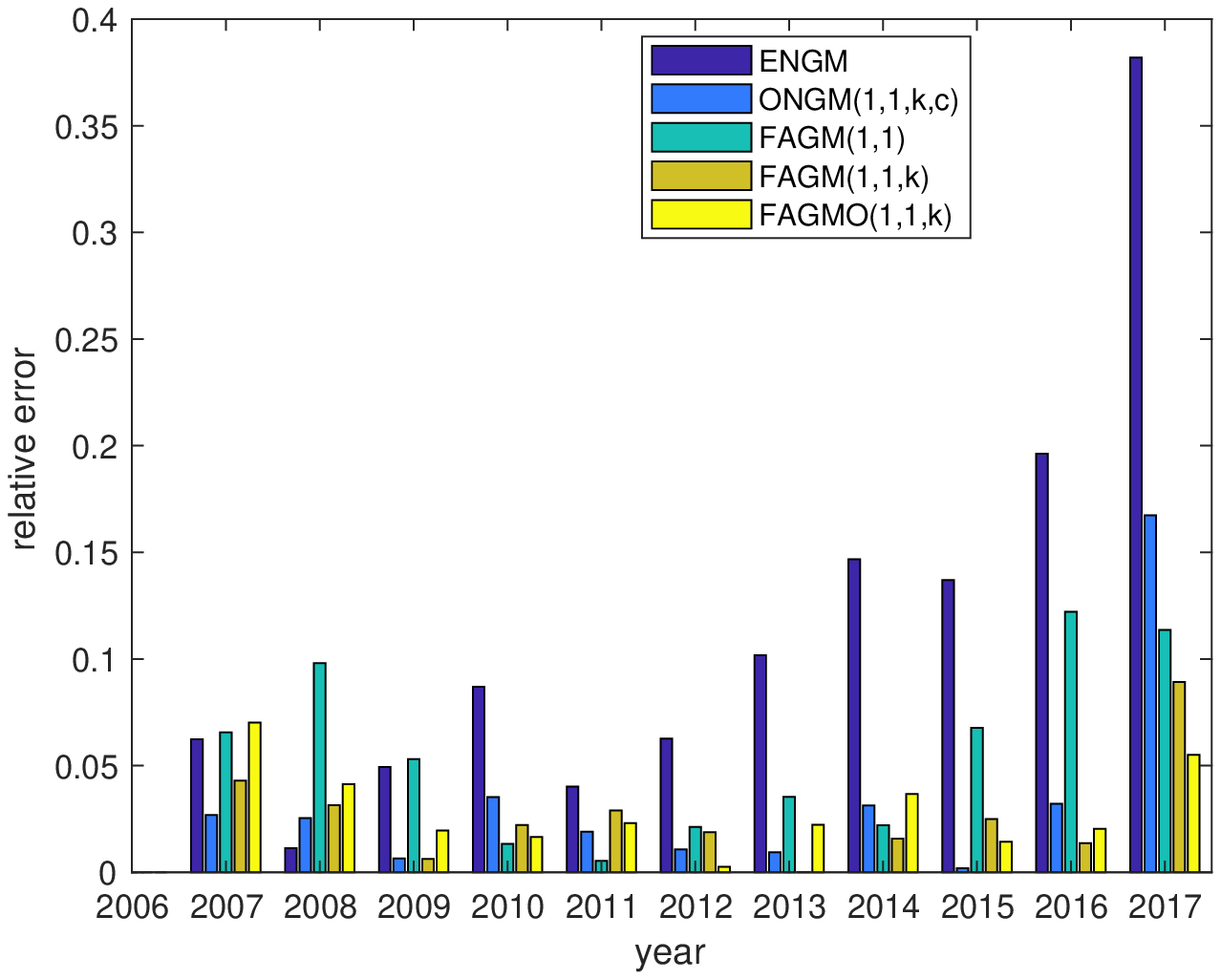}
\end{minipage}%
\begin{minipage}[c]{0.5\textwidth}
\centering
\includegraphics[height=5.8cm,width=6.2cm]{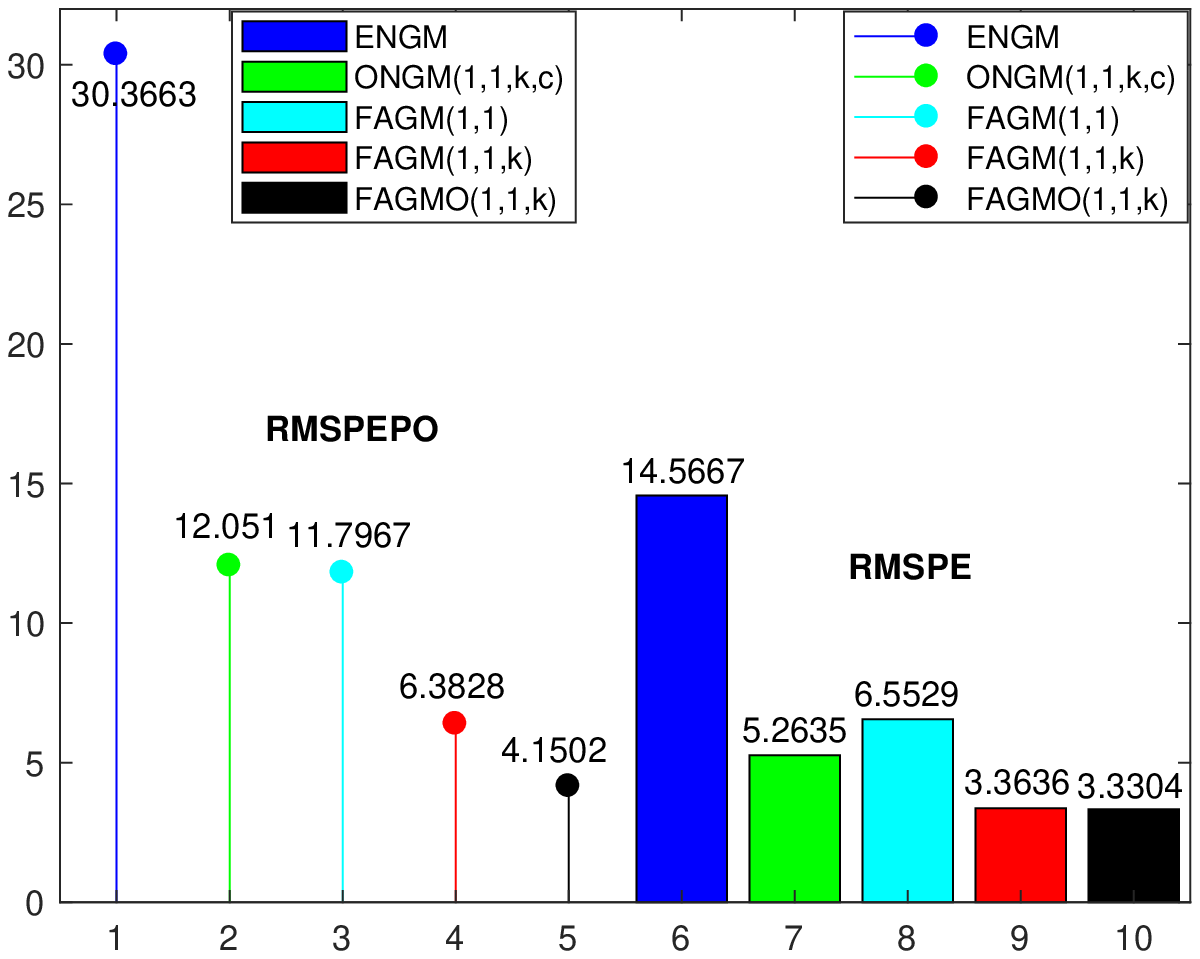}
\end{minipage}
\caption{Errors among five grey models for nuclear energy
consumption} \label{fig9:error}
\end{figure}
\newpage
\subsection{Further discussions}
\label{subsec:discussions}

As demonstrated by the case study, the novel FAGMO(1,1,$k$) model
outperforms other grey models. Moreover, it should be noted that in
this paper we only conduct short-term forecasting, while it is well
known that several existing energy models can perform long-term
forecasting, such as LEAP, TIMES and NEMS. We will discuss the
difference between our model and these models further, following a
very brief introduction to such models.

$\bullet$ LEAP (long-range energy alternatives planning system)
\cite{schnaars1987houLRRP, Dong2017APS} is a scenario-based energy
environment modelling tool for climate change mitigation and energy
policy analysis. It can be applied to examine energy production and
consumption, as well as resource extraction in all sectors. The model
studies the effects of various factors on energy consumption under
different scenarios given an objective. LEAP is generally used for
forecasting studies of between 20 and 50 years.

$\bullet$ TIMES (The Integrated MARKAL-EFOM System)
\cite{Shi2016modellingAE, Zhang2016timesAE} is an evolution of
MARKAL, which was developed by the Energy Technology Systems
Analysis Programme of the IEA. It
combines technical engineering and economic approaches,
and uses linear programming to produce a least-cost energy system
under numerous user-specified constraints. The software is used
to analyse energy, economic and environmental issues at different
levels over several decades.

$\bullet$ NEMS (National Energy Modeling System)
\cite{Gabriel2001theOR, Soroush2017aIJSSD} is a long-standing US
government policy model, which computes equilibrium fuel prices and
quantities for the US energy sector. NEMS is used to model the
demand side explicitly; in particular, to determine consumer
technology choices in the residential and commercial building
sectors.

These models can perform long-term energy consumption projections.
However, they may require a large amount of data, such as population
growth, GDP, urbanisation, energy policies and energy strategies. In
numerous practical situations, it is very difficult to obtain
complete information because of time and cost limitations. The grey
prediction model is an efficient method for conducting accurate
forecasting with at least four samples. Compared to the major energy
models, the grey model is an effective choice for predicting China's
nuclear energy consumption.

This paper collected 12 samples of China's nuclear energy
consumption from the {\it BP Statistical Review of World Energy
2018}. Thus, the LEAP, TIMES and NEMS models are all inapplicable
owing to poor information. By employing the grey system theory and
actual data from the $11^{\rm th}$ and the $12^{\rm th}$ Five-Year
Plans, the FAGMO(1,1,$k$) model was constructed. It can be observed
in Table \ref{table4:Predictionresults} that the prediction value of
FAGMO(1,1,$k$) is 123.3723 Mtoe in 2020, which is larger than the
84.6318 Mtoe provided in the {\it BP energy outlook 2018}. The main
reasons for this are as follows.

{\romannumeral 1)} It is infeasible to consider factors such as
energy policies and China's energy strategies, which affect the
current situation of China's nuclear energy consumption, in our
proposed model because the FAGMO(1,1,$k$) model is univariate.
However, the forecasting models of institutions including BP, the
IEA and APEC are based on widely collected data. Furthermore, grey
models are mainly used for short-term forecasting in the calculation
process, such as \cite{Feng2012ForecastingESPBEPP, Cui2013AAPM,
duan2018forecasting, Wu2018usingE, Chen2008ForecastingCNSNS}.
Therefore, the forecasting results are relatively acceptable,
reflecting the growth trend of future nuclear energy consumption in
China.

{\romannumeral 2)} In China's nuclear energy market, 38 nuclear
power reactors are in operation, 19 nuclear power reactors are under
construction and more are to be constructed by the end of 2016. This is
the reason for the increase in nuclear energy consumption in recent years. However, no
new nuclear projects have been approved for construction in 2016.
Moreover, the State Council approved new safety rules and a nuclear
power development plan following Japan's Fukushima Daiichi crisis in
2011. These factors have also resulted in a slight slowdown in
China's nuclear energy consumption.

In the future, nuclear energy could provide an important alternative
to fossil fuels such as coal and oil, and its proportion of the
total primary energy consumption will increase yearly. Based on
our forecasting results using the FAGMO(1,1,$k$) model, the future
nuclear energy consumption of China will increase rapidly if no certain
restrictions are placed thereon. This implies that higher management and
technical levels are necessary to meet the safety and quality requirements.
Therefore, China's government and policy makers
should pay additional attention to the safety and quality issues of
nuclear energy to achieve long-term, environmentally
friendly and low-carbon energy goals and lay the foundation for
the sustainable development of China's energy and economy.

%
%

\section{Conclusions}
\label{sec:conclu}

By applying the grey modelling technique and parameter optimisation
method, the fractional FAGMO(1,1,$k$) model was proposed to predict
China's nuclear energy consumption of the $13^{\rm th}$ Five-Year
Plan, based on the updated data from 2006 to 2015. The forecasting
results provide the growth trend of the future nuclear energy
consumption of China, and also offer a guideline for policymaking
and project planning.

It can be observed that FAGMO$(1,1,k)$ is quite easy to use, with
satisfactory accuracy in short-term nuclear consumption forecasting.
For long-term prediction, its error will be larger because only 10
samples are used for modelling. This study is expected to be able to
forecast the energy consumption of other countries that share
similar patterns of economic development and energy consumption
structures, among others. Furthermore, the optimised method applied
to improve the FAGM(1,1,$k$) model can be used to improve other
first-order grey models, such as NGBM(1,1), GMC(1,$n$) and
RDGM(1,$n$). These are possible extensions and suggested directions
for our future research.
\section*{Acknowledgments}

This research was supported by the National Natural Science
Foundation of China (No. 71771033), Longshan academic talent
research supporting program of SWUST (No. 17LZXY20), Doctoral
Research Foundation of Southwest University of Science and
Technology (No. 15zx7141, 16zx7140) and the Open Fund (PLN201710)
of the State Key Laboratory of Oil and Gas Reservoir Geology and
Exploitation (Southwest Petroleum University).

\section*{References}

\bibliography{greybib}

\end{document}